\documentclass{article}

     \usepackage[nonatbib,preprint]{neurips_2019}

\usepackage[utf8]{inputenc} 
\usepackage[T1]{fontenc}    
\usepackage{hyperref}       
\usepackage{url}            
\usepackage{booktabs}       
\usepackage{amsfonts}       
\usepackage{nicefrac}       
\usepackage{microtype}      

\usepackage{amsbsy}
\usepackage{amssymb}
\usepackage{amsmath}
\usepackage{amsthm}
\usepackage{bbm}
\usepackage{graphicx}
\usepackage{setspace}
\usepackage{esint}
\usepackage{comment}
\usepackage{color}
\usepackage{url}
\usepackage{hyperref}
\usepackage{float}
\usepackage{subfig}
\usepackage[capitalize]{cleveref}

\usepackage[linesnumbered,ruled,vlined]{algorithm2e}

\usepackage{color-edits}
\addauthor{vs}{blue}
\addauthor{md}{red}

\newcommand{\cR}{{\cal R}}
\newcommand{\cF}{{\cal F}}
\newcommand{\cH}{{\cal H}}
\newcommand{\cY}{{\cal Y}}
\newcommand{\cA}{{\cal A}}
\newcommand{\cZ}{{\cal Z}}

\newcommand{\kibitz}[2]{\ifnum\Comments=1{\color{#1}{#2}}\fi}

\newcommand{\E}{\mathbb{E}}

\newcommand{\R}{\mathbb{R}}

\renewcommand{\Pr}{\ensuremath{\mathrm{Pr}}}

\newcommand{\ba}{\begin{array}}
\newcommand{\ea}{\end{array}}
\newcommand{\bs}{\begin{align}\begin{split}\nonumber}
\newcommand{\bsnumber}{\begin{align}\begin{split}}
\newcommand{\es}{\end{split}\end{align}}

\newtheorem{assumption}{ASSUMPTION}

\newcommand{\ldot}[2]{\langle #1, #2 \rangle}

\newcommand{\Var}{\ensuremath{\text{Var}}}

\newcommand{\norm}[1]{\|{#1}\|} %

\def\balign#1\ealign{\begin{align}#1\end{align}}
\def\balignat#1\ealign{\begin{alignat}#1\end{alignat}}
\def\bitemize#1\eitemize{\begin{itemize}#1\end{itemize}}
\def\benumerate#1\eenumerate{\begin{enumerate}#1\end{enumerate}}
\newenvironment{talign}
 {\csname align\endcsname}
 {\endalign}
\def\balignt#1\ealignt{\begin{talign}#1\end{talign}}%
\newcommand{\cX}{\ensuremath{\mathcal{X}}}

\newcommand{\diag}{\mathtt{diag}}
\newcommand{\eye}{\mathbb{I}}

\newtheorem{theorem}{Theorem}
\newtheorem{corollary}[theorem]{Corollary}
\newtheorem{lemma}[theorem]{Lemma}
\newtheorem{definition}{Definition}
\newtheorem{example}{Example}
\newtheorem{remark}{Remark}

\title{Semi-Parametric Efficient Policy Learning with Continuous Actions}

\author{%
  Mert Demirer \\
  MIT\\
  \texttt{mdemirer@mit.edu} \\
  \And
    Vasilis Syrgkanis\\
  Microsoft Research\\
  \texttt{vasy@microsoft.com}
  \AND
  Greg Lewis \\
  Microsoft Research\\
  \texttt{glewis@microsoft.com} \\
  \And
    Victor Chernozhukov \\
  MIT\\
  \texttt{vchern@mit.edu} \\
}

\begin{document}

\maketitle

\begin{abstract}
We consider off-policy evaluation and optimization with continuous action spaces. We focus on observational data where the data collection policy is unknown and needs to be estimated. We take a semi-parametric approach where the value function takes a known parametric form in the treatment, but we are agnostic on how it depends on the observed contexts. We propose a doubly robust off-policy estimate for this setting and show that off-policy optimization based on this estimate is robust to estimation errors of the policy function or the regression model. Our results also apply if the model does not satisfy our semi-parametric form, but rather we measure regret in terms of the best projection of the true value function to this functional space. Our work extends prior approaches of policy optimization from observational data that only considered discrete actions. We provide an experimental evaluation of our method in a synthetic data example motivated by optimal personalized pricing and costly resource allocation.
\end{abstract}

\section{Introduction}

We consider off-policy evaluation and optimization with continuous action spaces from observational data, where the data collection (logging) policy is unknown. We take a semi-parametric approach where we assume that the value function takes a known parametric form in the treatment, but we are agnostic on how it depends on the observed contexts/features. In particular, we assume that:
\begin{equation}\label{eqn:main-ass}
    V(a, z) = \langle \theta_0(z), \phi(a, z) \rangle
\end{equation}
for some known feature functions $\phi$ but unknown functions $\theta_0$. We assume that we are given a set of $n$ observational data points $(x_{1},..., x_{n}) $ that consist of i.i.d copies of the random vector $x=(y, a, z)\in \cY\times \cA\times \cZ$, such that $\E[y\mid a, z]=V(a, z)$.\footnote{In most of the paper, we can allow for the case where $z$ is endogenous, in the sense that $\E[y \mid a, z] = V(a, z) + f_0(z)$. In other words, the noise in the random variable $y$ can be potentially correlated with $z$. However, we assume that conditional on $z$, there is no remaining endogeneity in the choice of the action in our data. 
The latter is typically referred to as conditional ignorability/exogeneity \cite{imbens2015causal}.} 

Our goal is to estimate a policy $\hat{\pi}: \cZ\rightarrow \cA$ from a space of policies $\Pi$ that achieves good regret:
\begin{equation}
    \sup_{\pi\in \Pi} \E[V(\pi(z), z)] - \E[V(\hat{\pi}(z), z)] \leq R(\Pi, n)
\end{equation}
for some regret rate that depends on the policy space $\Pi$ and the sample size $n$. 

The semi-parametric value assumption allows us to formulate a doubly robust estimate $V_{DR}$ of the value function, from the observational data, which depends on first stage regression estimates of the coefficients $\theta_0(z)$ and the conditional covariance of the features $\Sigma_0(z)=\E[\phi(a,z)\phi(a,z)^T\mid z]$. The latter is the analogue of the propensity function when actions are discrete. Our estimate is doubly robust in that it is unbiased if either $\theta_0$ or $\Sigma_0$ is correct. Then we optimize this estimate:
\begin{equation}
    \hat{\pi} = \sup_{\pi \in \Pi} V_{DR}(\pi) 
\end{equation}

\paragraph{Main contributions.} We show that the double robustness property implies that our objective function satisfies a Neyman orthogonality criterion, which in turn implies that our regret rates depend only in a second order manner on the estimation errors on the first stage regression estimates of the functions $\theta_{0}, \Sigma_0$. Moreover, we prove a regret rate whose leading term depends on the variance of the difference of our estimated value between any two policy values within a ``small regret-slice'' and on the entropy integral of the policy space. We achieve this with a computationally efficient variant of the empirical risk minimization (ERM) algorithm (of independent interest) that uses a validation set to construct a preliminary policy and use it to regularize the policy computed on the training set. Hence, we manage to achieve variance-based regret bounds without the need for variance or moment penalization \cite{maurer2009empirical,swaminathan2015counterfactual,foster2019orthogonal} used in prior work and which can render a computationally tractable policy learning problem, non-convex. We also show that the asymptotic variance of our off-policy estimate (which governs the scale of the leading regret term) is asymptotic minimax optimal, in the sense that it achieves the semi-parametric efficiency lower bound.


\paragraph{Robustness to mis-specification.} Notably, our approach provides meaningful guarantees even when our semi-parametric value
function assumption is violated. Suppose that the true value function does not take the form of Equation~\eqref{eqn:main-ass}, but rather takes some other form $V_0(a, z)$. Then one can consider the projection of this value function onto the forms of Equation~\eqref{eqn:main-ass}, as:
\begin{equation} \label{eq:proj}
    \theta_{p}(z) = \arg\inf_{\theta} \E\left[ (V_0(a,z) - \langle \theta(z), \phi(a, z) \rangle)^2 \mid z\right]
\end{equation}
where the expectation is taken over the distribution of observed data. 
Then our approach takes the interpretation of achieving good regret bounds with respect to this best linear semi-parametric approximation. This is an alternative to the kernel smoothing approximation proposed by \cite{swaminathan2015counterfactual} in contextual bandit setting, as a regret target, and related to \cite{kallus2018policy}. If there is some rough domain knowledge on the form of how the action affects the reward, then our semi-parametric approximate target should achieve better performance when the dimension of the action space is large, as the bias of kernel methods will typically incur an exponential in the dimension bias.

\paragraph{Double robustness.} In cases where the collection policy is known, our doubly robust approach can be used for variance reduction via fitting first stage regression estimates to the policy value, whilst maintaining unbiasedness. Thus we can apply our approach to improve regret in the counterfactual risk minimization framework \cite{swaminathan2015counterfactual}, \cite{kallus2018policy} and as a variance reduction method in contextual bandit algorithms with continuous actions \cite{swaminathan2015counterfactual}.

\paragraph{Related Literature.} Our work builds on the recent work at the intersection of semi-parametric inference and policy learning from observational data. The important work of \cite{athey2017efficient} analyzes the binary treatments and infinitesimal nudges to continuous treatments. They also take a doubly robust approach so as to obtain regret bounds whose leading term depends on the semi-parametric efficient variance and the entropy integral and which is robust to first stage estimation errors. The problem we study in this paper is different in that we consider optimizing over continuous action spaces, rather than infinitesimal nudges, under semi-parametric functional form. This assumption is without loss of generality if treatment is binary or multi-valued. Hence, our results are a generalization of binary treatments to arbitrary continuous actions spaces, subject to our semi-parametric value assumption. In fact we show formally in the Appendix how one can recover the result of \cite{athey2017efficient} for the binary setting, from our main regret bound. In turn our work builds on a long line of work on policy learning and counterfactual risk minimization \cite{qian2011performance,zhao2012estimating,zhou2017residual,athey2017efficient,kitagawa2018should,zhou2018offline,beygelzimer2009offset,dudik2011doubly,swaminathan2015counterfactual,kallus2018policy,krishnamurthy}. Notably, the work of \cite{zhou2018offline} extends the work of \cite{athey2017efficient} to many discrete actions, but only proves a second moment based regret bound, which can be much larger than the variance. Our setting also subsumes the setting of many discrete actions and hence our regularized ERM offers an improvement over the rates in \cite{zhou2018offline}. \cite{foster2019orthogonal} formulates a general framework of statistical learning with a nuisance component. Our method falls into this framework and we build upon some of the results in \cite{foster2019orthogonal}. However, for the case of policy learning the implications of \cite{foster2019orthogonal} provide a variance based regret only when invoking second moment penalization, which can be intractable. We side-step this need and provide a computationally efficient alternative. Finally, most of the work on policy learning in machine learning assumes that the current policy (equiv. $\Sigma_0(z)$) is known. Hence, double robustness is used mostly as a variance reduction technique. Even for this literature, as we discuss above, our method can be seen an alternative of recent work on policy learning with continuous actions \cite{kallus2018policy,krishnamurthy} that makes use of non-parametric kernel methods.

Our work also connects to the semi-parametric estimation literature in econometrics and statistics. Our model is an extension of the partially linear model which has been extensively studied in the econometrics \cite{engle1986semiparametric, robinson1988root}. By considering context-specific coefficients (random coefficients) and modeling a value function that is non-linear in treatment, we substantially extend the partially linear model. \cite{wooldridge2004estimating, graham2018semiparametrically} studied a special case of our model where output is linearly dependent on treatment given context, with the aim of estimating the average treatment effect. \cite{graham2018semiparametrically} constructed the doubly robust estimator and showed its semi-parametric efficiency under the linear-in-treatment assumption. We extend their results to a more general functional form and use the double-robustness property and semi-parametric efficiency for policy evaluation and optimization rather than treatment effect estimation.
Our work is also connected to the recent and rapidly growing literature on the orthogonal/locally robust/debiased estimation literature \cite{chernozhukov2018double, chernozhukov2016locally, van2011targeted}.

\section{Orthogonal Off-Policy Evaluation and Optimization} \label{sec:ortho}

Let $\hat{\theta}$ be a first stage estimate of $\theta_0(z)$, which can be obtained by minimizing the square loss:
\begin{equation} \label{direct}
    \hat{\theta} = \arg\inf_{\theta\in \Theta} \E_n\left[ \left(y - \langle \theta(z), \phi(a, z) \rangle\right)^2 \right]
\end{equation}
where $\Theta$ is an appropriate parameter space for the parameters $\theta(z)$. Let $\Sigma_0(z)$ denote the conditional covariance matrix:
\begin{equation*}
    \Sigma_0(z) = \E[\phi(a, z)\, \phi(a, z)^T \mid z]
\end{equation*}
This is the analogue of the propensity model in discrete treatment settings. An estimate $\hat{\Sigma}(z)$ can be obtained by running a multi-task regression problem for each entry to the matrix, i.e.:
\begin{equation}  \label{covar}
    \hat{\Sigma}_{ij} = \arg\inf_{\Sigma_{ij}\in {\cal S}_{ij}} \E\left[ (\phi_i(a, z)\, \phi_j(a, z) - \Sigma_{ij}(z))^2 \right]
\end{equation}
where ${\cal S}_{ij}$ is some appropriate hypothesis space for these regressions.
Finally, the doubly robust estimate of the off-policy value takes the form:
\begin{equation}\label{eqn:dr-value}
    V_{DR}(\pi) = \E_n\left[v_{DR}(y, a, z; \pi)\right]
\end{equation}
where:
\begin{align} \label{eq: orthogonal}
    v_{DR}(y, a, z; \pi) =~& \ldot{\theta_{DR}(y, a, z)}{\phi(\pi(z), z)}\\
    \theta_{DR}(y, a, z) =~& \hat{\theta}(z) + \hat{\Sigma}(z)^{-1}\, \phi(a, z)\, (y - \langle \hat{\theta}(z), \phi(a, z)\rangle)
\end{align}
The quantity $\theta_{DR}(y, a, z)$ can be viewed as an estimate of $\theta_0(z)$, based on a single observation. In fact, if the matrix $\hat{\Sigma}$ was equal to $\Sigma_0$, then one can see that $\theta_{DR}(y, a, z)$ is an unbiased estimate of $\theta_0(z)$. Our estimate $v_{DR}$ also satisfies a doubly robust property, i.e. it is correct if either $\hat{\theta}$ is unbiased or $\hat{\Sigma}^{-1}$ is unbiased (see  Appendix~\ref{app:dr} for a formal statement). Finally, we will denote with $\theta_{DR}^0(y, a, z)$ the version of $\theta_{DR}$, where the nuisance quantities $\theta$ and $\Sigma$ are replaced by their true values, and correspondingly define $v_{DR}^0(y, a, z; \pi)$.
We perform policy optimization based on this doubly robust estimate:
\begin{equation} \label{ERM}
\hat{\pi} = \arg\max_{\pi \in \Pi} V_{DR}(\pi)
\end{equation}
Moreover, we let $\pi_*^0$ be the optimal policy:
\begin{equation}
    \pi_*^0 = \arg\max_{\pi\in \Pi} V(\pi)
\end{equation}

\begin{remark}[Multi-Action Policy Learning]
A special case of our setup is the setting where the number of actions is finitely many. This can be encoded as $a \in \{e_1, \ldots, e_n\}$ and $\phi(a, z)=a$. In that case, observe that the covariance matrix becomes a diagonal matrix: $\Sigma_0(z) = \diag(p_1(z), \ldots, p_n(z))$, with 
$p_i(z)=\Pr[a=e_i \mid z]$. In this case, we simply recover the standard doubly robust estimate that combines the direct regression part with the inverse propensity weights part, i.e.:
\begin{equation*}
    \theta_{DR, i}(y, a, z) = \hat{\theta}_i(z) + \frac{1}{\hat{p}_i(z)}\, 1\{a=e_i\}\, (y - \hat{\theta}_i(z))
\end{equation*}
Thus our estimator is an extension of the doubly robust estimate from discrete to continuous actions.
\end{remark}

\begin{remark}[Finitely Many Possible Actions: Linear Contextual Bandits]
Another interesting special case of our approach is a generalization of the linear contextual bandit setting. In particular, suppose that there is only a finite (albeit potentially large) set of $N>p$  possible actions $A=\{a_1,\ldots, a_N\}$ and $a_i\in \R^p$. However, unlike the multi-action setting, where these actions are the orthonormal basis vectors, in this setting, each action $a_i\in A$, maps to a feature vector $\phi_i(z):=\phi(a_i, z)$. Then the reward $y$ that we observe satisfies $\E[y\mid z, a]= \ldot{\theta(z)}{\phi(a, z)}$. This is a generalization of the linear contextual bandit setting, in which the coefficient vector $\theta(z)$ is a constant parameter $\theta$ as opposed to varying with $z$.
In this case observe that: $\Sigma_0(z)=\sum_{i=1}^N p_i(z)\, \phi_i(z)\, \phi_i(z)^T= UDU^T $, i.e. it is the sum of $N$ rank one matrices where $D = \diag(p_1(z), \ldots, p_n(z))$, $p_i(z)=\Pr[a=a_i\mid z]$ and $U = [\phi_i(z), \dots,\phi_N(z) ]$ The doubly robust estimate of the parameter takes the form:
%
%
\begin{equation*}
    \theta_{DR}(y, a, z) = \hat{\theta}(z) +  (UDU^T)^{-1} \, \phi(a, z)\, (y - \langle \hat{\theta}(z), \phi(a, z)\rangle)
\end{equation*}
This approach leverages the functional form assumption to get an estimate that avoids a large variance that depends on the number of actions $N$ but rather mostly depends on the number of parameters $p$. This is achieved by sharing reward information across actions.
\end{remark}

\begin{remark}[Linear-in-Treatment Value]
Consider the case where the value is linear in the action $\phi(a, z)=a\in \R^p$.
In this case observe that: $\Sigma_0(z)=\Var(a\mid z)$. For instance, suppose that we assume that experimentation is independent across actions in the observed data. Then $\Sigma_0(z)=\diag(\sigma_1^2(z), \ldots, \sigma_p^2(z))$, where $\sigma_i^2=\Var(a_i\mid z)$. Then the doubly robust estimate of the parameter takes the form:
\begin{equation}
    \theta_{DR, i}(y, a, z) = \hat{\theta}_i(z) +  \frac{a_i}{\hat{\sigma}_i^2(z)}\, (y - \ldot{\hat{\theta}(z)}{a})
\end{equation}
\end{remark}

\section{Theoretical Analysis}

Our main regret bounds are derived for a slight variation of the ERM algorithm that we presented in the preliminary section. In particular, we crucially need to augment the ERM algorithm with a ``validation'' step, where we split our data into a training and validation step, and we restrict attention to policies that achieve small regret on the training data, while still maintaining small regret on the validation set. This extra modification enabled us to prove variance based regret bounds and is reminiscent of standard approaches in machine learning, like $k$-fold cross-validation and early stopping, hence could be of independent interest.

\begin{algorithm}[H]
The inputs are given by the sample of data $ S =(x_1, \dots, x_n)$, which we randomly split in two parts $S_1, S_2$.  Moreover, we randomly split $S_2$ into validation and training samples $S_2^v$ and  $S_2^t$.

Estimate the nuisance functions $\hat \theta (z)$ and  $\hat \Sigma (z)$ using Equations (\ref{direct}) and (\ref{covar}) on $S_1$.

Use the output of Step 2 to construct the doubly robust moment in Equation (\ref{eqn:dr-value}) on $S_2^v$. Run ERM given in Equation (\ref{ERM}) over policy space $\Pi_1 $ on $S_2^v$ to learn a policy function $\pi_1$.

Use the output of Step 3 to construct a function class $\Pi_2 $ defined as 
\begin{equation*}
    \Pi_2 = \{\pi \in \Pi : \E_{S^v}[v_{DR}(y, a, z; \pi_1) - v_{DR}(y, a, z; \pi)] \leq \mu_n\}
\end{equation*}
for some  $\mu_n$ and $\E_{S^v} $ denotes the empirical expectation over the validation sample.

Use the output of Step 1 to construct the doubly robust moment in Equation (\ref{eqn:dr-value}) on $S_2^t$. Run a constrained ERM on $S_2^t$ over $\Pi_2$.
\caption{Out-of-Sample Regularized ERM with Nuisance Estimates}\label{alg:implementation} 
\end{algorithm}
We note that we present our theoretical results for the simpler case where the nuisance estimates are trained on a separate split of the data. However, our results qualitatively extend to the case where we use the cross-fitting idea of \cite{chernozhukov2018double} (i.e. train a model on one half and predict on the other and vice versa).

\paragraph{Regret bound.} To show the properties of this algorithm, we first show that the regret of the doubly robust algorithm is impacted in a second order manner by the errors in the first stage estimates. We will also make the following preliminary definitions. For any function $f$ we denote with $\|f\|_2 = \sqrt{\E[f(x)^2]}$, the standard $L^2$ norm and with $\|f\|_{2,n}=\sqrt{\E_n[f(x)^2]}$ its empirical analogue. Furthermore, we define the empirical entropy of a function class $H_2(\epsilon, \cF, n)$ as the largest value, over the choice of $n$ samples, of the logarithm of the size of the smallest empirical $\epsilon$-cover of $\cF$ on the samples with respect to the $\|\cdot\|_{2,n}$ norm. Finally, we consider the empirical entropy integral:
\begin{equation}
    \kappa(r, \cF) = \inf_{\alpha\geq 0}\left\{ 4\alpha + 10\int_{\alpha}^{r} \sqrt{\frac{\cH_2(\epsilon, \cF, n)}{n}} d\epsilon\right\},
\end{equation}
Our statistical learning problem corresponds to learning over the function space:
\begin{equation}
    \cF_{\Pi} = \{ v_{DR}(\cdot;\pi): \pi\in \Pi\}
\end{equation}
where the data is $x=(y,a,z)$. We will also make a very benign assumption on the entropy integral:
\begin{assumption}
The function class $\cF_{\Pi}$ satisfies that for any constant $r$, $\kappa(r, \cF)\rightarrow 0$ as $n\rightarrow \infty$.
\end{assumption}

\begin{theorem}[Variance-Based Oracle Policy Regret]\label{thm:main-regret}
Suppose that the nuisance estimates satisfy that their mean squared error is upper bounded w.p. $1-\delta/2$ by $h_{n,\delta}^2$, i.e. w.p. $1-\delta/2$ over the randomness of the nuisance sample:
\begin{equation}
    \max\left\{\E[(\hat{\theta}(z)-\theta_0(z))^2], \E[\|\hat{\Sigma}(z)-\Sigma_0(z)\|_{Fro}^2]\right\} \leq h_{n,\delta}^2
\end{equation}
Let $r=\sup_{\pi\in \Pi} \sqrt{\E[v_{DR}(z;\pi)^2]}$ and $\mu_n = \Theta\left( \kappa(r, \cF_{\Pi}) + r \sqrt{\frac{\log(1/\delta)}{n}}\right)$. Moreover, let 
\begin{equation}
    \Pi_*(\epsilon)=\{\pi\in \Pi: V(\pi_*^0) - V(\pi)\leq \epsilon\},
\end{equation}
denote an $\epsilon$-regret slice of the policy space. 
Let $\epsilon_n=O(\mu_n + h_{n,\delta}^2)$ and
\begin{equation}
    V_2^0 = \sup_{\pi,\pi'\in \Pi_*(\epsilon_n)} \Var(v_{DR}^0(x; \pi)-v_{DR}^0(x; \pi')) 
\end{equation}
denote the variance of the difference between any two policies in an $\epsilon_n$-regret slice, evaluated at the true nuisance quantities. Then the policy $\pi_2$ returned by the out-of-sample regularized ERM, satisfies w.p. $1-\delta$ over the randomness of $S$:
\begin{align}
    V(\pi_*^0) - V(\pi_2) =~& O\left( \kappa(\sqrt{V_2^0}, \cF_{\Pi}) +  \sqrt{\frac{V_2^0\, \log(1/\delta)}{n}} + h_{n,\delta}^2\right)
\end{align}
Expected regret is $O\left(\kappa(\sqrt{V_2^0}, \cF_{\Pi}) +  \sqrt{\frac{V_2^0}{n}} + h_{n}^2\right)$, with $h_n^2$ is expected MSE of nuisance functions.
\end{theorem}

We provide a proof of this Theorem in Appendix~\ref{app:main-theorem}. The regret result contains two main contributions: 1) first the impact of the nuisance estimation error is of second order (i.e. $h_{n,\delta}^2$ instead of $h_{n,\delta}$), 2) the leading regret term depends on the variance of small-regret policy differences and the entropy integral of the policy space. The first property stems from the Neyman orthogonality property of the doubly robust estimate of the policy. The second property stems from the out-of-sample regularization step that we added to the ERM algorithm. Typically, we will have $h_{n,\delta}^2=o(1/\sqrt{n})$ and thereby this term is of lower order than the leading term. Moreover, for many policy spaces $\kappa(0, \cF_{\Pi})=0$, in which case we see that if the setting satisfies a ``margin'' condition (i.e. the best policy is better by a constant $\Delta$ margin), then eventually the variance of small regret policies is $0$, since it only contains the best policy. In that case, our bound leads to fast rates of $\log(n)/n$ as opposed to $1/\sqrt{n}$, since the leading term vanishes (similar to the $\log(n)/n$ achieved in bandit problems with such a margin condition). 

Dependence on the quantity $V_2^0$ is quite intuitive: if two policies have almost equivalent regret up to a $\mu_n$ rate, then it will be very easy to be mislead among them if one has much higher variance than the other. For some classes of problems, the above also implies a regret rate that only depends on the variance of the optimal policy (e.g. when all policies with low regret have a variance that is not much larger than the variance of the optimal policy. In Appendix~\ref{app:variogram} we show that the latter is always the case for the setting of binary treatment studied in \cite{athey2017efficient} and therefore applying our main result, we recover exactly their bound for binary treatments.


\paragraph{Semi-parametric efficient variance.} Our regret bound depends on the variance of our doubly robust estimate of the value function. One then wonders if there are other estimates of the value function that could achieve better variance than $V_{DR}(\pi)$. However, we show that at least asymptotically and without further assumptions on the functions $\theta_{0}(z)$ and $\Sigma_{0}(z)$, this cannot be the case. In particular, we show that our estimator achieves what is known as the semi-parametric efficient variance limit for our setting. More importantly for our regret result, this is also true for the semiparametric efficient variance of the policy differences. This is the case in our main setup; where the model is mis-specified and only a projection of the true value; and even if we assume that our model is correct, but make the extra assumption of homoskedasticity, i.e., the conditional variance of residuals of outcomes $y$ do not depend on $(a, z)$.

\begin{theorem}[Semi-parametric Efficiency]\label{thm:semi-param-efficiency}
If the model is mis-specified, i.e, $V_{0}(a,z)\neq V(a,z)$ the asymptotic variance of $V_{DR}(\pi)$ is equal to the semi-parametric efficiency bound for the policy value $ \langle \theta_p(z), \pi(z) \rangle $ defined in Equation (\ref{eq:proj}). If the model is correctly specified, $V_{DR}(\pi)$ is semi-parametrically efficient under homoskedasticity, i.e. $\E[ (y - V(a,z))^2 \mid a,z] = \E[ (y - V(a,z))^2]$.
\end{theorem}

We provide a proof for the value function, but this result also extends to the difference of values. We conclude the section by providing concrete examples of rates for policy classes of interest.

\begin{example}[VC Policies]
As a concrete example, consider the case when the class $\cF_{\Pi}$ is a VC-subgraph class of VC dimension $d$ (e.g. the policy space has small VC-dimension or pseudo-dimension), and let $S=\E[\sup_{\pi} v_{DR}(x;\pi)^2]$. Then Theorem 2.6.7 of \cite{VanDerVaartWe96} shows that: $\cH_2(\epsilon, \cF_{\Pi}, n)=O( d(1+\log(S/\epsilon)))$ (see also discussion in Appendix~\ref{app:variogram}). This implies that
\begin{align*}
\kappa(r, \cF_{\Pi}) = O\left(\int_0^{r} \sqrt{d(1+\log(S/\epsilon))}d\epsilon\right) = O\left(r \sqrt{d}\sqrt{1+\log(S/r)}\right).
\end{align*}
Hence, we can conclude that regret is $O\left(  \sqrt{V_2^0(1+\log(S/V_2^0))}\sqrt{\frac{d}{n}} +  \sqrt{\frac{V_2^0\, \log(1/\delta)}{n}} + h_{n,\delta}^2\right)$.
For the case of binary action policies (as we discuss in Appendix~\ref{app:variogram}) this result recovers the result of \cite{athey2017efficient} for binary treatments up to constants and extends it to arbitrary action spaces and VC-subgraph policies.
\end{example}

\begin{example}[High-Dimensional Index Policies]
As an example, we consider the class of policies, characterized by a constant number of $\ell_1$ or $\ell_0$-bounded linear indices:
\begin{align}
    \Pi_1 = \{z \rightarrow \Gamma(\ldot{\beta_1}{z}, \ldots, \ldot{\beta_d}{z}): \beta_i \in \R^p, \|\beta_i\|_1 \leq s\}
\end{align}
where $\Gamma: \R^d \rightarrow \R^m$ is a fixed $L$-Lipschitz function of the indices, with $d, m$ constants, while $p>>n$ (and similarly for $\Pi_0$, where use $\|\beta_i\|_0\leq s$). Assuming $v_{DR}(y, a, z;\pi)$ is a Lipschitz function of $\pi(z)$ and since $\Gamma$ is a Lipschitz function of $\ldot{\beta}{z}$, we have by a standard multi-variate Lipschitz contraction argument (and since $d$, $m$ are constants), that the entropy of $\cF_{\Pi}$ is of the same order as the maximum entropy of each of the linear index spaces: $B_1:=\{z \rightarrow \ldot{\beta_i}{z}: \beta \in \R^p, \|\beta_i\|_1 \leq s\}$. Moreover, by known covering arguments (see e.g. \cite{zhang2002covering}, Theorem~3) that if $\|z\|_{\infty}\leq 1$, then:
$\cH_2(\epsilon, B, n) = O\left(\frac{s^{2}\log(d)}{\epsilon^{2}}\right)$.
Thus we get $\kappa(r, \cF_\Pi)=O\left( s\, \log(n)\, \sqrt{\frac{\log(d)}{n}} + \frac{r}{n}\right)$,
which leads to regret $O\left( s\, \log(n)\, \sqrt{\frac{\log(d)}{n}} +  \sqrt{\frac{V_2^0\, \log(1/\delta)}{n}} + h_{n,\delta}^2\right)$.
In this setting, we observe that the policy space is too large for the variance to drive the asymptotic regret. There is a leading term that remains even if the worse-case variance of policies in a small-regret slice is $0$. Intuitively this stems from the high-dimensionality of the linear indices, which introduces an extra dimension of error, namely bias due to regularization. On the contrary, for exactly sparse policies $B_0:=\{z \rightarrow \ldot{\beta_i}{z}: \beta \in \R^p, \|\beta_i\|_0 \leq s\}$, we have that since for any possible support the entropy at scale $\epsilon$ is at most $O(s\log(1/\epsilon))$, we can take a union over all ${ p\choose s}\leq{}\left(\frac{ep}{s}\right)^{s}$ possible sparse supports, which implies $\cH_2(\epsilon,\cF_{\Pi},n)= O(s\left(\log(d/s) + \log(1/\epsilon)\right)$. Thus $\kappa(r, \cF_{\Pi}) = O\left(r \sqrt{\log(1/r)} \sqrt{\frac{s\log(d/s)}{n}}\right)$,
leading to policy regret similar to the VC classes: $O\left(\sqrt{V_2^0\log(1/V_2^0)} \sqrt{\frac{s\log(d/s)}{n}} +  \sqrt{\frac{V_2^0\, \log(1/\delta)}{n}} + h_{n,\delta}^2\right)$.
\end{example}

\begin{remark}[Instrumental Variable Estimation]
Our main regret results extend to the instrumental variables settings where  treatments are endogenous but we have a vector of instrumental variables $w$ satisfying
\begin{align*}
\E[ (y - \langle {\theta}_{0}(z), \phi (a, z)\rangle)w \mid z ] = 0,
\end{align*}
and $\Sigma_0^{I}(z) = \E[ w \phi(a, z)^T \mid z] $ is invertible. Then we can use the following doubly robust moment
\begin{align*}
\theta_{DR,I}(y, a, z, w) =~& \hat{\theta}(z) + \hat{\Sigma}^{I}(z)^{-1}\, w (y - \langle \hat{\theta}(z), \phi(a, z)\rangle).
\end{align*}
\end{remark}

\begin{remark}[Estimating the First Stages]
Bounds on first stage errors as a function of sample complexity measures for the first stage hypotheses spaces can be obtained by standard results on the MSE achievable by regression problems (see e.g. \cite{rakhlin2017empirical, wainwright2019}). Essentially these are bounds for the regression estimates $\hat \theta(z)$ and $\hat \Sigma(z)$, as a function of the complexity of their assumed hypothesis spaces. Since the latter is a standard statistical learning problem that is orthogonal to our main contribution, we omit technical details. Since the square loss is a strongly convex objective the rates achievable for these problems are typically fast rates on the MSE (e.g. $h_{n,\delta}^2$ is of the order $1/n$ for the case of parametric hypothesis spaces, and typically $o(1/\sqrt{n})$ for reproducing kernel Hilbert spaces with fast eigendecay (see e.g. \cite{wainwright2019})). Thus the term $h_{n,\delta}^2$ is of lower order. For instance, the required rates for the term $h_{n,\delta}^2$ to be of second order in our regret bounds are achievable if these nuisance regressions are $\ell_1$-penalized linear regressions and several regularity assumptions are satisfied by the data distribution, even when the dimension $p$ of $z$ is growing with $n$.
\end{remark}

\paragraph{Extension: Semi-Bandit Feedback}
Suppose that our value function takes the form:
$V(a, z) = \phi(a, z)^T\, \Theta_0(z)\, \phi(a, z)$,
where $\Theta_0(z)$ is a $p\times p$ matrix and we observe semi-bandit feedback, i.e. we observe a vector $Y$ s.t.: $
    \E[ Y \mid a, z] = \Theta_0(z)^T \phi(a, z)$.
Then we can apply our DR approach to each coordinate of $Y$ separately.
\begin{equation*}
    V_{DR}(\pi) = \E_n\left[ \phi(\pi(z), z)^T\, \left(\hat{\Theta}(z) + \hat{\Sigma}(z)^{-1}\, \phi(a, z)\, (Y^T - \phi(a, z)^T \hat{\Theta}(z))\right) \, \phi(\pi(z), z) \right]
\end{equation*}
All the theorems in this section extend to this case, which will prove useful in our pricing application where $a$ is the price of a set of products and $Y$ is the vector of observed demands for each product.

\section{Application: Personalized Pricing} 

Consider the personalized pricing of a single product. The objective is the revenue:
\begin{align*}
    V(p, z) = p\, (a(z) - b(z)\, p)
\end{align*}
where $b(z)\geq \gamma > 0$ and $a(z) + b(z)p $ gives the unknown, context-specific demand function. We assume that we observe an unbiased estimate $d$ of demand:
\begin{equation*}
    \E[d\mid z, p] = a(z) - b(z)\, p
\end{equation*}
We want to optimize over a space of personalized pricing policies $\Pi$. If, for instance, the observational policy was homoskedastic (i.e. the exploration component was independent of the context $z$), we show in Appendix \ref{sec:pricing} that doubly robust estimators for $a(z)$ and $b(z)$ are
\begin{align*}
    a_{DR}(z) 
    =~& \hat{a}(z) + \left(1 + \hat{g}(z)\frac{\hat{g}(z) - p}{\hat{\sigma}^2} \right)\, (d - \hat{a}(z) - \hat{b}(z)\, p) \\
    b_{DR}(z) =~& \hat{b}(z) + \frac{p - \hat{g}(z)}{\hat{\sigma}^2} (d - \hat{a}(z) - \hat{b}(z)\, p)
\end{align*}
where $g(z)=\E[p\mid z]$ and the variance $\sigma^2$. Thus, in this example, we only need to estimate the mean treatment policy $\E[p\mid z]$ and the variance $\sigma^2$.

\begin{figure}[t]
\centering
\subfloat[Policy Evaluation]{\includegraphics[width=0.9 \textwidth,
height=5cm]{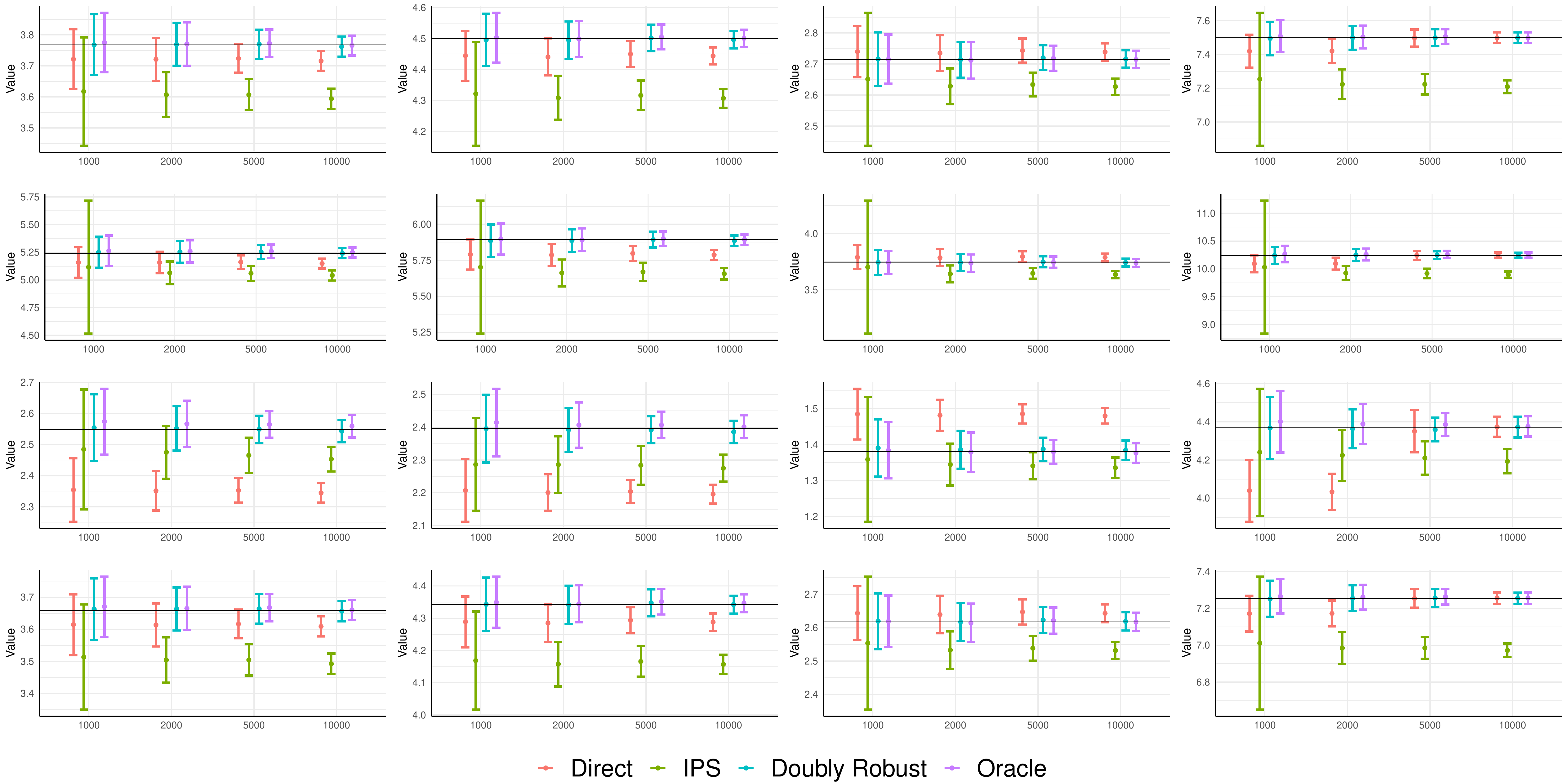}}

\vspace{-0.4cm}

\subfloat[Regret]{\includegraphics[width=0.9 \textwidth, height=2.5cm]{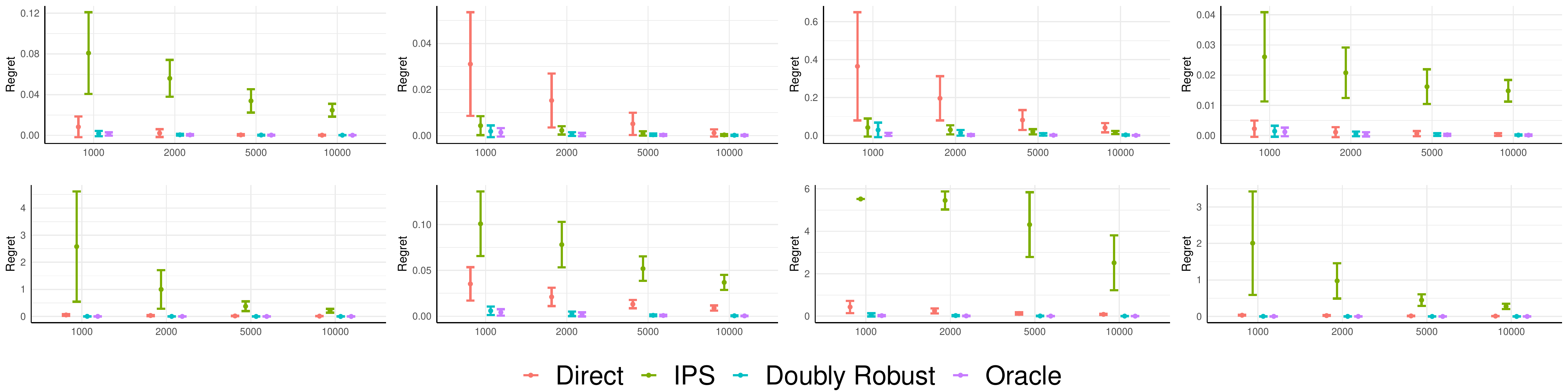}}
\caption{(a) Black line shows the true value of the policy and each line shows the mean and standard deviation of the value of the corresponding policy over 100 simulations. (b) Each line shows the mean and standard deviation of regret over 100 simulations. The top half reports the regret for a constant policy, the bottom half reports regret for a linear policy.}
\label{exp} 
\end{figure}

\paragraph{Experimental evaluation.} We empirically evaluate our framework on the personalized pricing application with synthetic data. In particular, we use simulations to assess our estimator's ability to evaluate and optimize personalized pricing functions. To do this, we compare the performance of our doubly robust estimator with (1) Direct estimator, $\langle \hat \theta (z),  \phi(a,z) \rangle$, (2) Inverse propensity score estimator \footnote{  $\theta_{IPS}(y, a, z) =  \hat{\Sigma}(z)^{-1}\, \phi(a, z)\, y$}, (3) Oracle orthogonal estimator, $v_{DR}^{o} (x, \pi)$.

\paragraph{Data Generating Process.} 
Our simulation design considers a sparse model. We assume that there are $k$ continuous context variables distributed uniformly $z_{i} \sim U(1,2)$ for $i=1, \dots, k$ but only $l$ of them affects demand. Let $\bar z = 1/l (z_i + \dots + z_l)$. Price $p$ and demand $d$ are generated as
 $   x  \sim \mathcal{N} (\bar z, 1),  d= a(\bar z) -   b (\bar z) x + \epsilon \  \text{and} \ \epsilon \sim  \mathcal{N} (0, 1) $. We consider four functional forms for the demand model: (i) (Quadratic) $ a (z) =2z^2,  b (z) = 0.6z$, (ii) (Step) $ a (z) = 5 \{ z<1.5\} + 6 \{ z>1.5\}, b (z) = 0.7 \{ z<1.5\} + 1.2 \{ z>1.5\}$, (iii) (Sigmoid) $ a (z) = 1/(1+\exp(z)) + 3, b(z) = 2/(1+\exp(z)) +0.1 $, (iv) (Linear) $ a (z) = 6 z, b (z) = z  $ 

These functions and the data generating process ensure that the conditional expectation function of demand given $z$ is non-negative for all $z$, the observed prices are positive with high probability, and the optimal prices are in the support of the observed prices. In each experiment, we generate 1000, 2000, 5000, and 10000 data points, and report results over 100 simulations. We estimate the nuisance functions using 5-fold cross-validated lasso model with polynomials of degrees up to 3 and all the two-way interactions of context variables. We present the results for two regimes: (i) Low dimensional with $k=2, l=1$, (ii) High dimensional with $k=10, l=3$.
\paragraph{Policy Evaluation.}  For policy evaluation we consider four pricing functions: (i) Constant, $ \pi(z)=1 $, (ii) Linear, $ \pi(z)=z $, (iii) Threshold, $ \pi(z)= 1 + 1\{ z>1.5 \}$, (iv) Sin, $ \pi(z)= \text{sin}(z)$. The results for the low dimensional regime are summarized in Figure \ref{exp}(a), where each row and column corresponds to a different demand function and a policy function, respectively\footnote{The results are very similar for the high dimensional model which are reported in Figure \ref{fig:exp2}(a) in the appendix.}. The results show that, as expected, our the performance of our method is very similar to the oracle estimator and achieves a significantly better performance than the direct and inverse propensity score methods, which suffer from large biases. These results also support our claim that the asymptotic variance of the doubly robust estimate is the same as the variance of the oracle method. It is also important to point out that we obtain similar performances across two different regimes.
\paragraph{Regret.} To investigate the regret performance of our method, we consider a constant pricing function, $\pi(z) = \gamma$ and a linear policy $\pi(z) = (\gamma_1 z_1 + \dots + \gamma_k z_k)$. We compute the optimal pricing functions in these two function spaces and report the distribution of regret in Figure \ref{fig:exp2}(b) under the low dimensional regime and in Appendix \ref{sec:figures} under the high dimensional regime. Across the four demand functions and two pricing functions we consider, our method achieves small regrets, comparable to the oracle method. The direct and inverse propensity methods, depending on the demand function, yield large regrets. 
\subsection{Quadratic Model} 
Finally, we consider the same simulation exercise under the assumption that an unbiased estimate of revenue rather than demand is observed. Since revenue depends on the $p^2$ the model is now quadratic
\begin{align*}
 r &= a(z) p -   b (z) p^2 + \epsilon
\end{align*}
For the data generating process we use the same functions $a(z)$ and $b(z)$ as in the personalized pricing example \footnote{We provide the calculation of the doubly robust estimators for this example in Appendix \ref{sec:pricing}.}. Figures \ref{fig:exp3} and \ref{fig:exp4} in Appendix \ref{sec:figures} summarize results for policy evaluation and optimization. The overall performance of our doubly robust estimator is similar to the demand model, and it performs better the direct model. One important difference to note is that when the sample size is small, we observe some finite sample biases for some function classes.

\begin{figure}[t]
\centering
\subfloat[Policy Evaluation]{\includegraphics[width=0.9 \textwidth,
height=5cm]{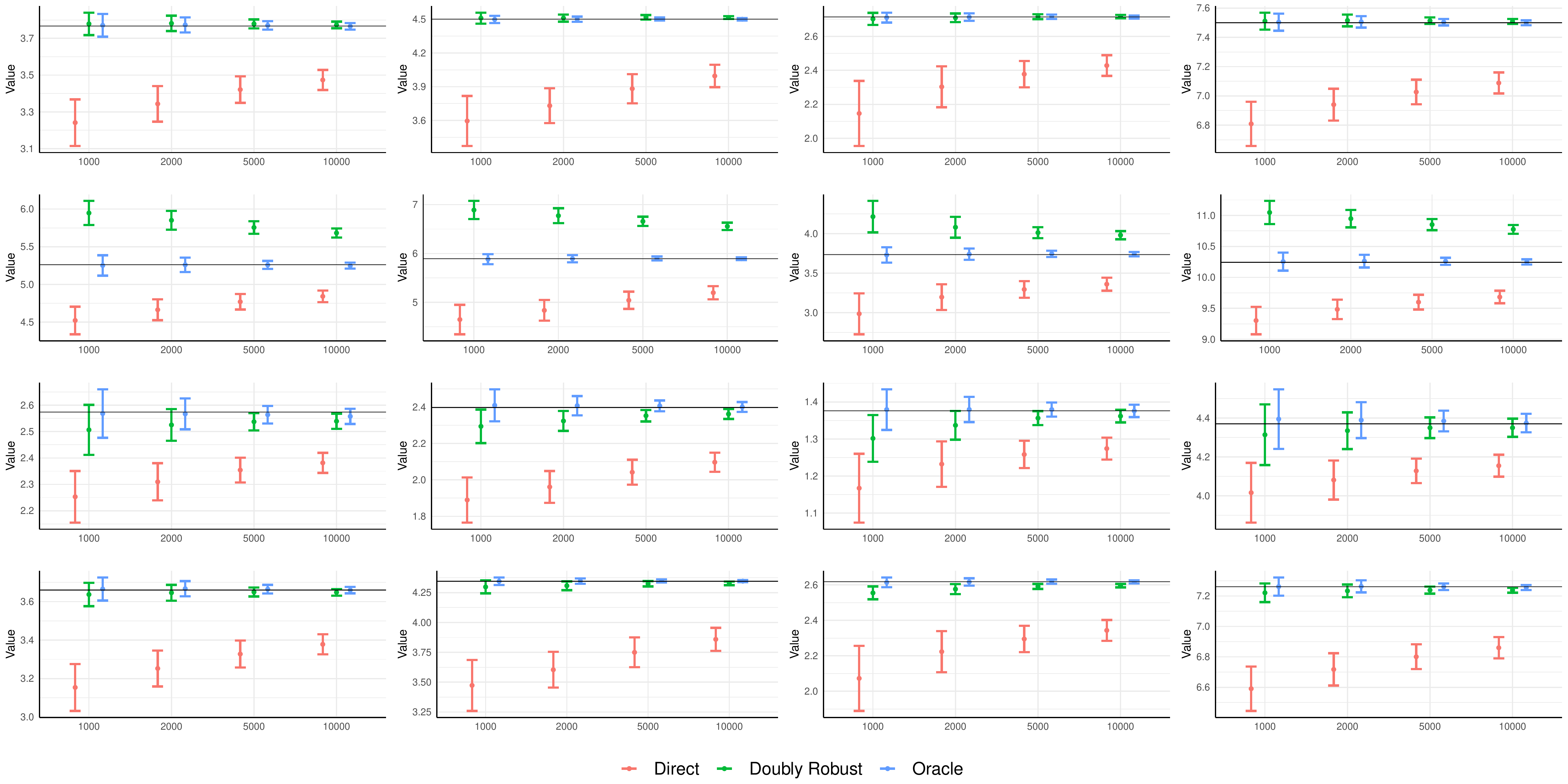}}

\subfloat[Regret]{\includegraphics[width=0.9 \textwidth, height=2.5cm]{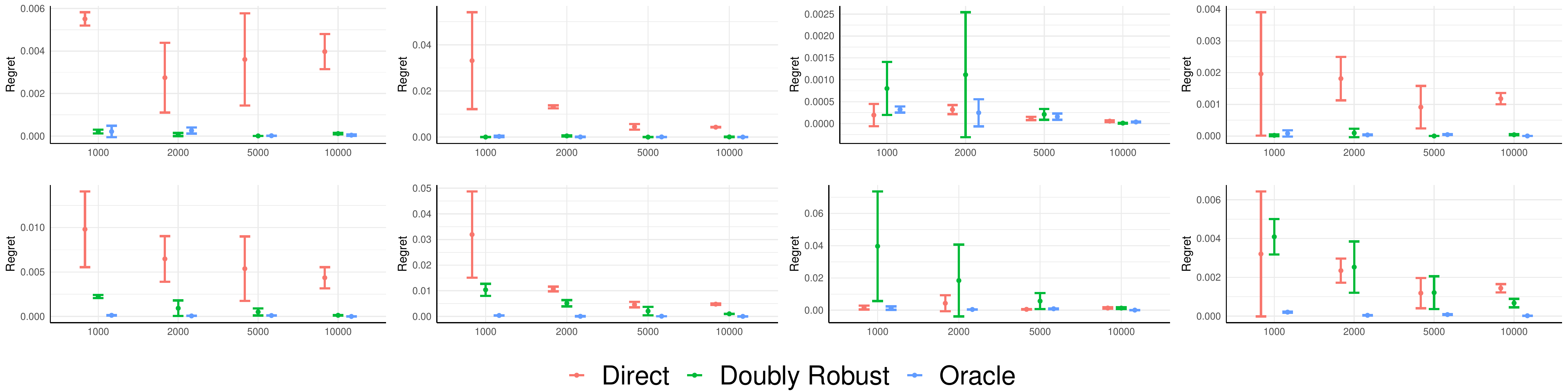}}
\caption{Quadratic, Low Dimensional Regime: (a) Black line shows the true value of the policy, each line shows the mean and standard deviation of the value of the corresponding policy over 100 simulations. (b) Each line shows the mean and standard deviation of regret over 100 simulations. We omit the results for the inverse propensity score method since they are too large to report together with the other estimates in the high dimensional regime.}
\label{fig:exp3} 
\end{figure}

\section{Application: Costly Resource Allocation} \label{sec:resource}

Motivated by a resource allocation scenario, we also analyze experimentally the special case where $\phi(a, z)=a$. Consider the case where we have $p$ possible tasks to invest in, and we have investment costs. Each task yields a return on investment that is a linear function of the investment, but an unknown function $\theta(z)$ of the context $z$. Moreover, to maintain an investment portfolio of $\pi(z)$ we need to pay a known cost $C(\pi(z))$. Given a policy space $\Pi: {\cal Z}\rightarrow \R^p$, our goal is to optimize:
\begin{equation}
    \sup_{\pi \in \Pi} \E\left[ \ldot{\theta(z)}{\pi(z)} - C(\pi(z)) \right]
\end{equation}

This falls into our framework, if we treat the offset part as of the form $\ldot{\theta_0(z)}{C(\pi(z))}$ but with a known $\theta_0(z)=1$. So in that case we simply consider $\theta_{DR,0}(z)=\theta_0(z)=1$. Then applying our framework we optimize:
\begin{equation}
    \sup_{\pi\in \Pi} \E_n\left[\ldot{\theta_{DR}(z)}{\pi(z)} - C(\pi(z))\right]
\end{equation}

In the case of quadratic costs $C(\pi(z))=\frac{\lambda}{2} \|\pi(z)\|_2^2$, then this boils down to exactly optimizing a square loss objective, since:
\begin{equation}
     \inf_{A} \E_n\left[\|\theta_{DR}(z)/\lambda - \pi(z)\|^2\right] \Leftrightarrow 
    \sup_{A} \E\left[\ldot{\theta_{DR}(z)}{\pi(z)}\right] - \frac{\lambda}{2}\E_n\left[\|\pi(z)\|_2^2\right]
\end{equation}
Thus policy optimization reduces to a multi-task regression problem where we are trying to predict $\theta_{DR}(z)/\lambda$ from $z$.\footnote{The above reasoning extends to heterogeneous costs across tasks e.g. $C(\pi(z)) = \sum_i c_i \pi_i(z)^2$. In this case the label of the $i$-th task of the multi-task regression problem is $\theta_{DR,i}(z)/c_i$ and we need to perform a weighted multi-task regression where the weight on the square loss for task $i$ is equal $c_i$.}

We can consider sparse linear policies:
\begin{equation}
    \Pi = \{z \rightarrow Az: \|A\|_{11} := \sum_i \|\alpha_i\|_1 \leq s\}
\end{equation}
where $\alpha_i$ corresponds to the $i$-th row of matrix $A$. In this case our problem reduces to the MultiTask Lasso problem where the label is $\theta(z)/\lambda$.

\paragraph{Experimental Evaluation.} For experimental evaluation we consider a model with two tasks, $a_1$ and $a_2$:
\begin{align*}
    y = a(z) a_1 + b(z) a_2 + \epsilon
\end{align*}
We use the same distributions and functions, $a(z)$ and $b(z)$, given above for the pricing application. To estimate the optimal allocation and its regret, we run a 5-fold cross validated MultiTask Lasso algorithm and set $\lambda=1$. We report the distribution of return on investment obtained from different models in Figure (\ref{resource}). The results suggest that doubly robust method achieves a significantly lower regret than the direct method in both regimes and its performance is similar to the oracle method \footnote{For comparison, the value achieved by best-in-class policy is 22.2 in low dimensional regime and ? in high dimensional regime. We omit the inverse propensity score regrets since they are too large to report together with other estimates}.
\begin{figure}[t]
\centering
\subfloat[Low Dimensional Regime]{\includegraphics[width=0.9 \textwidth,
height=2cm]{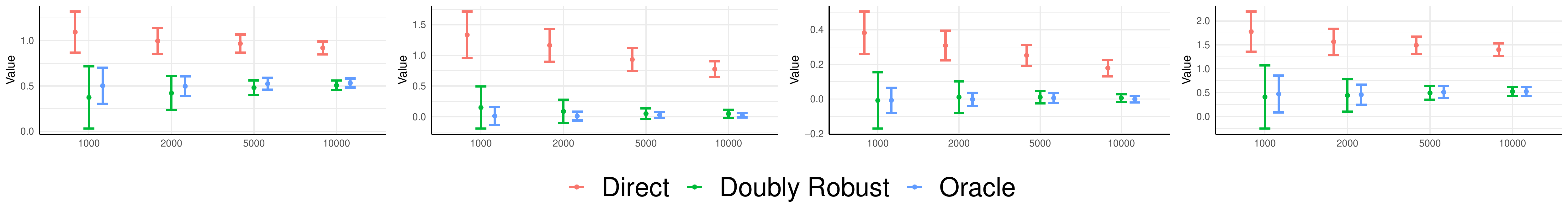}}

\subfloat[High Dimensional Regime]{\includegraphics[width=0.9 \textwidth, height=2cm]{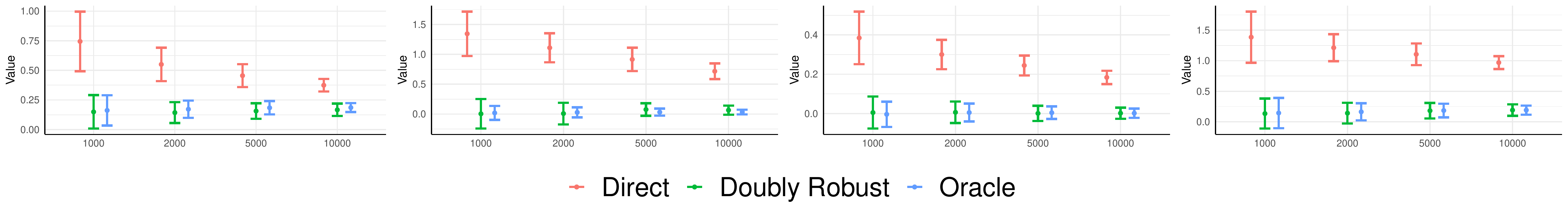}}
\caption{Costly Resource Allocation: Each line shows the mean and standard deviation of regret over 100 simulations.}
\label{resource} 
\end{figure}

\bibliographystyle{plain}
\bibliography{references}

\newpage

\appendix

\section{Proof of Universal Orthogonality Lemma}
We first start by defining a sufficient condition for the notion of \emph{universal orthogonality} of a loss function, as defined by \cite{foster2019orthogonal}. A loss function $L(\pi; h)=\E[\ell(x, \pi(z); h(z))]$ is universally orthogonal with respect to $h$ if for any $\pi\in \Pi$:
\begin{equation}
    \E[\nabla_{h(z), \pi(z)} \ell(x, \pi(z); h_0(z)) \mid z]=0
\end{equation}
where $h_0$ is the true value of the nuisance parameter $h$.

\begin{lemma}\label{lem:univ-orthogonality}
The loss function $L(\pi; h) = -\E[\ldot{\theta_{DR}(y, a, z)}{\phi(\pi(z), z)}]$ is universally orthogonal with respect to $h=(\theta, \Sigma)$.
\end{lemma}
\begin{proof}
We show that the population loss function that corresponds to the doubly robust estimate, satisfies the universal orthogonality condition. For simplicity of notation let $K(z)=\Sigma(z)^{-1}$. Then the population loss is:
\begin{align*}
    V_{DR}^0(\pi; \hat{\theta}, \Sigma^{-1}) = \E\left[ \left\langle \hat{\theta}(z) + \Sigma^{-1}(z)\, \phi(a,z)\, (y - \langle \hat{\theta}(z), \phi(a,z)\rangle), \phi(\pi(z),z) \right\rangle \right]
\end{align*}
Let:
\begin{align*}
    \beta(a,z, \xi, K) =  \xi + K \phi(a,z)\, (y - \langle \xi, \phi(a,z)\rangle)
\end{align*}
Observe that:
\begin{align*}
    V_{DR}^0(\pi; \hat{\theta}, \Sigma^{-1}) = \E\left[ \left\langle \beta(a, \hat{\theta}(z), \Sigma^{-1}(z)), \phi(\pi(z),z) \right\rangle \right]
\end{align*}
To show universal orthogonality it suffices to show that:
\begin{align*}
    \E\left[ \nabla_{\xi, K}\beta(a, z, \theta_0(z), \Sigma_0^{-1}(z)) \mid z\right]=~&0
\end{align*}
This follows easily by simple algebraic manipulations:
\begin{align*}
    \E\left[ \nabla_{\xi}\beta(a,z,\theta_0(z), \Sigma_0^{-1}(z)) \mid z\right]=~&\E\left[\eye-\Sigma_0^{-1}(z)\, \phi(a,z) \phi(a,z)^T \mid z\right]\\
    =~&
    \eye-\Sigma_0^{-1}(z) \E\left[\phi(a,z) \phi(a,z)^T \mid z\right]=\eye-\Sigma_0^{-1}(z)\, \Sigma_0(z)=0
\end{align*}
and
\begin{align*}
    \E\left[ \nabla_{K_{ij}}\beta(a, \theta_0(z), \Sigma_0^{-1}(z)) \mid z\right]=~\E\left[\phi_j(a,z)\, (y - \langle \theta_0(z), \phi(a,z)\rangle)\mid z\right]
\end{align*}
Now observe that since $\theta_0(z)$ is
the minimizer of the conditional squared loss, taking the first order condition implies:
\begin{align*}
     \E[(V_0(a, z)  & - \ldot{\theta_0(z)}{\phi(a,z)})\, \phi(a,z  )\mid z]=0  \Longleftrightarrow \\ & \E[V_0(a, z)\, \phi(a,z)\mid z] = \E[\ldot{\theta_0(z)}{\phi(a,z)})\, \phi(a,z)\mid z]
\end{align*}
Moreover:
\begin{align*}
    \E[y\, \phi(a,z)\mid z] = \E[ \E[y\mid a, z]\, \phi(a,z)] = \E[V_0(a,z)\, \phi(a,z) ]
\end{align*}
Combining the two yields:
\begin{equation*}
\E\left[\phi(a,z)\, (y - \langle \theta_0(z), \phi(a,z)\rangle)\mid z\right] = 0
\end{equation*}
which implies orthogonality with respect to $K$.
\end{proof}

\section{Proof of Main Regret Theorem~\ref{thm:main-regret}}
\label{app:main-theorem}

We first consider an arbitrary empirical loss minimization problem of the form:
\begin{equation}
    f_n = \arg\min_{f\in \cF} \E_n[f(x)] := \frac{1}{n}\sum_{i=1}^n f(x_i)
\end{equation}
where $x_i\in \cX$ are i.i.d. drawn from an unknown distribution and $\cX$ is an arbitrary data space. Throughout the section we will assume that: $\sup_{f\in \cF} |f(x)|\leq 1$. All the results can be generalized to the case of $\sup_{f\in \cF} |f(x)|\leq R$, for some arbitrary $R$, by simply first re-scaling the losses, and then invoking the theorems of this section. 

We will also make the following preliminary definitions. For any function $f$ we denote with $\|f\|_2 = \sqrt{\E[f(x)^2]}$, the standard $L^2$ norm and with $\|f\|_{2,n}=\sqrt{\E_n[f(x)^2]}$ its empirical analogue. The localized Rademacher complexity is the defined as:
\begin{equation}
    \cR(r, \cF) = \E_{\epsilon, x_{1:n}}\left[\sup_{f\in \cF: \|f\|_2\leq r} \frac{1}{n} \sum_{i=1}^n \epsilon_i\, f(x_i)\right]  
\end{equation}
where $\epsilon_i$ are independent Rademacher variables that take values $\{-1, 1\}$ with equal probability. 

Furthermore, we define the empirical entropy of a function class $H_2(\epsilon, \cF, n)$ as the largest value, over the choice of $n$ samples, of the logarithm of the size of the smallest empirical $\epsilon$-cover of $\cF$ on the samples with respect to the $\|\cdot\|_{2,n}$ norm. Finally, we consider the empirical entropy integral defined as:
\begin{equation}
    \kappa(r, \cF) = \inf_{\alpha\geq 0}\left\{ 4\alpha + 10\int_{\alpha}^{r} \sqrt{\frac{\cH_2(\epsilon, \cF, n))}{n}} d\epsilon\right\},
\end{equation}
Throughout this section we will make the following benign assumption that essentially makes the problem \emph{learnable}:

\textbf{ASSUMPTION 1.} \textit{The function class satisfies that for any constant $r$, $\kappa(r, \cF)\rightarrow 0$ as $n\rightarrow \infty$}

We will use the following theorems from the prior work of \cite{foster2019orthogonal} as a starting point as they are formalized in manner convenient for our problem.
\begin{theorem}[Foster, Syrgkanis \cite{foster2019orthogonal}, Theorem 4]
\label{thm:regret-foster}
Consider any function class $\cF: \cX\rightarrow [-1, 1]$ and let $f_n$ be the outcome of the constrained ERM. Pick any $f_*\in \cF$ and let
$r=\sup_{f\in \cF}\|f-f_*\|_2$.  Then for some constants $C_1, C_2$ and for any $\delta>0$, w.p. $1-\delta$: 
\begin{align*}
\E[f_n(x) - f_*(x)]
\leq~& C_1\,\left( \cR(r, \cF-f^*) + r \sqrt{\frac{\log(1/\delta)}{n}} + \frac{\log(1/\delta)}{n}\right)\\
\leq~& C_1\, C_2\, \left( \kappa(r, \cF)  + r \sqrt{\frac{\log(1/\delta)}{n}} + \frac{\cH_2(r, \cF, n)}{n} +\frac{\log(1/\delta)}{n}\right).
\end{align*}
\end{theorem}

\begin{lemma}[Foster, Syrgkanis \cite{foster2019orthogonal}, Lemma 4]
\label{lem:ucon-foster}
Consider a function class $\cF: \cX\rightarrow [-1,1]$ and pick any $f_*: \cX\rightarrow [-1, 1]$ (not necessarily in $\cF$). Moreover, let:
\begin{equation}
 Z_n(r) = \sup_{f\, \in\, \cF: \|f-f^*\|_{2}\leq r} \left| \E_n[f(x) - f_*(x)] - \E[f(x) - f_*(x)] \right|   
\end{equation}
Then for some constant $C_3$ and for any $\delta>0$, w.p. $1-\delta$:
\begin{equation*}
  Z_n(r) \leq 
  C_3\,\left(\cR(r, \cF-f^*)+ r \sqrt{\frac{\log(1/\delta)}{n}} + \frac{\log(1/\delta)}{n}\right)  
\end{equation*}
\end{lemma}

Our goal is to replace $r$ in the latter Theorem with the worst-case variance of the functions $f\in \cF$ in a small ``regret''-ball around the optimal. We will achieve this by considering a slight modification of the ERM algorithm. In particular, we will split the data in half, and we will use one half as a \emph{regularization sample} and the other half as the \emph{training sample}. In particular, we will find the optimal function on the training sample, within the class of functions that also have relatively small regret on the regularization sample.

\paragraph{Out-of-Sample Regularized ERM} Consider the following algorithm:
\begin{itemize}
    \item We split the samples $S$ in two parts $S_1, S_2$ and let $\E_{n_1}[\cdot]$ and $\E_{n_2}[\cdot]$ denote the corresponding empirical expectations.
    \item We run ERM over $\cF$ on the first half and let $f_1$ be the outcome.
    \item Then we define the class of functions that have the constraint that they don't achieve much worse value than $f_1$ on the first half, i.e. we regularize policies based on their regret on the first half. More formally, for some constant $\mu_n$ to be defined later:
    \begin{equation}
        \cF_2 = \{f \in \cF: \E_{n_1}[f(x) - f_1(x)] \leq \mu_n\}
    \end{equation}
    \item Then we run constrained ERM on the second sample over the function space $\cF_2$:
    \begin{equation}
        f_2 = \arg\min_{f \in \cF_2} \E_{n_2}[f(x)]
    \end{equation}
\end{itemize}

\begin{theorem}[Variance-Based Regret]\label{thm:main-regret-general}
Let $f_*=\arg\min_{f\in \cF} \E[f(x)]$, $r=\sup_{f\in \cF} \|f\|_2$ and choose $\mu_n = C\,\left( \kappa(r, \cF) + r \sqrt{\frac{\log(6/\delta)}{n}} + \frac{\cH_2(r, \cF, n)}{n} +\frac{\log(6/\delta)}{n}\right)$, with $C=8\max\{C_1C_2, C_3C_2\}$. Then, w.p. $1-\delta$ over the sample $S$, the outcome $f_2$ of the Out-of-Sample Regularized ERM satisfies:
\begin{align}
    \E[f_2(x) - f_*(x)] =~& O\left( \kappa(\sqrt{V_2}, \cF_*(\mu_n)) +  \sqrt{\frac{V_2\, \log(3/\delta)}{n}}\right)
\end{align}
with: $\cF_*(\mu_n)=\{f\in \cF: \E[f(x)-f_*(x)]\leq \mu_n\}$ and $V_2 = \sup_{f\in \cF_*(\mu_n)} \Var(f(x)-f_*(x))$. Moreover, the expected regret, in expectation over the samples $S_1, S_2$ is also of order $O\left( \kappa(\sqrt{V_2}, \cF) +  \sqrt{\frac{V_2}{n}}\right)$.
\end{theorem}
\begin{proof}
First we argue that w.p. $1-\delta/6$, $f_*\in \cF_2$. By the choice of $\mu_n$ and Theorem~\ref{thm:regret-foster}, we know that w.p. $1-\delta/4$ over the randomness of sample $S_1$:
\begin{equation}
    \E[f_1(x) - f_*(x)] \leq \mu_n/2
\end{equation}
Moreover, by Lemma~\ref{lem:ucon-foster}, w.p. $1-\delta/6$ over the randomness of sample $S_1$:
\begin{align*}
    \sup_{f\in \cF} |\E_{n_1}[f(x)-f_*(x)] - \E[f(x)-f_*(x)]| \leq \mu_n/2
\end{align*}

Combining the latter two properties we have, w.p. $1-\delta/3$:
\begin{align*}
     |\E_{n_1}[f_*(x) - f_1(x)]| \leq |\E[f_*(x) - f_1(x)]| + \mu_n/2 \leq \mu_n
\end{align*}
Thus in this event, $f_*\in \cF_2$.

Applying Theorem~\ref{thm:regret-foster} for the last stage of the algorithm with function space $\cF_2$ and conditioning on the event that the first stage sample is such that $f_* \in \cF_2$, we have that with probability $1-\delta/3$ over the randomness of the second sample:
\begin{align*}
\E[f_2(x) - f_*(x)]
=~& C_1\, C_2\, \left( \kappa(r_2, \cF_2) + r_2 \sqrt{\frac{\log(3/\delta)}{n}} + \frac{\cH_2(r_2, \cF_2, n)}{n} + \frac{\log(3/\delta)}{n}\right)
\end{align*}
where $r_2 = \sup_{f\in \cF_2} \|f\|_2$. Thus by a union bound we get that with probability $1-2\delta/3$ over the randomness of both samples, the latter bound holds. 

Observe that for $f\in \cF_2$, by Lemma~\ref{lem:ucon-foster}, w.p. $1-\delta/6$ over the first sample:
\begin{align*}
    \sup_{f\in \cF} \left| \E_{n_1}[f(x) - f_1(x)] - \E[f(x)-f_1(x)]\right| \leq  2\, \sup_{f\in \cF}\left| \E_{n_1}[f(x)] - \E[f(x)]\right|
    \leq \mu_n/2
\end{align*}
Thus w.p. $1-\delta/6$, $\cF_2$ is a subset of the class:
\begin{align}
    \cF_2^0 = \{f\in \cF: |\E[f(x) - f_1(x)]| \leq \mu_n/2\}
\end{align}
Moreover, since $f_1$ has small regret, we know by the triangle inequality, for all $f \in \cF_2^0$, w.p. $1-\delta/3$:
\begin{align*}
     |\E[f(x) - f_*(x)]| \leq |\E[f(x) - f_1(x)]| + |\E[f_1(x) - f_*(x)]| \leq \mu_n
\end{align*}
Thus w.p. $1-\delta/3$, $\cF_2^0$ is in turn a subset of the function space:
\begin{align*}
    \cF_*(\mu_n) = \{f \in \cF: |\E[f(x) - f_*(x)]| \leq \mu_n \}
\end{align*}
which is a space of policies with regret at most $\mu_n$.

Thus we have that w.p. $1-\delta/3$ over the first sample:
\begin{align*}
    r_2^2=~& \sup_{f\in \cF_2} \E[(f(x) - f_*(x))^2] ~\leq~ \sup_{f\in \cF_*(\mu_n)} \E[(f(x) - f_*(x))^2]\\
    =~& \sup_{f\in \cF_*(\mu_n)} \left(\Var(f(x) - f_*(x)) + \E[f(x) - f_*(x))]^2\right)\\
    \leq~& \sup_{f\in \cF_*(\mu_n)} \Var(f(x) - f_*(x))  + \mu_n^2
\end{align*}
We thus have that:
\begin{equation}
    r_2 = \sqrt{\sup_{f\in \cF_*(\mu_n)} \Var(f(x) - f_*(x))}  + 2\mu_n = \sqrt{V_2}  + 2\mu_n
\end{equation}

Combining the latter with the regret bound for $f_2$ (excluding lower order terms in $n$) we have that w.p. $1-\delta$:
\begin{align*}
\E[f_2(x) - f_*(x)]
=~& O\left( \kappa(\sqrt{V_2} + 2\mu_n, \cF_*(\mu_n)) +  \sqrt{\frac{V_2\, \log(3/\delta)}{n}}\right)
\end{align*}
Moreover, using the concavity of the entropy integral with respect to its first argument, we have that:
\begin{align}
    \kappa(\sqrt{V_2} + 2\mu_n, \cF) \leq \kappa(\sqrt{V_2}, \cF_*(\mu_n)) + 2\mu_n \sqrt{\frac{H_2(\sqrt{V_2}, \cF, n)}{n}} 
\end{align}
Since $\kappa(r, \cF)\rightarrow 0$, we have that $\mu_n = o(1)$ and $H_2(\sqrt{V_2}, \cF, n)$ is a constant. Thus, the second term decays faster than $1/\sqrt{n}$ and hence is asymptotically negligible. Thus we get:
\begin{align*}
\E[f_2(x) - f_*(x)]
=~& O\left( \kappa(\sqrt{V_2}, \cF_*(\mu_n)) +  \sqrt{\frac{V_2\, \log(1/\delta)}{n}}\right)
\end{align*}
The expected regret bound follows by standard arguments by simply integrating the above high probability bound.
\end{proof}

Going back to our policy learning problem, let $x=(y, a, z)$ and:
\begin{equation}
    v_{DR}(x; \pi) = \ldot{\theta_{DR}(y, a, z)}{\phi(\pi(z),z)}
\end{equation}
be the doubly robust proxy value at every sample $x$ and policy $\pi$. Then we can apply this general theorem to the policy learning problem where, $x=(y, a, z)$ and function space:
\begin{equation}
    \cF_{\Pi} = \{- v_{DR}(\cdot;\pi): \pi\in \Pi\}
\end{equation}
Then Theorem~\ref{thm:main-regret-general} yields the following corollary:
\begin{corollary}[Variance-Based Policy Regret]\label{cor:policy-regret}
Let $\pi_*=\arg\max_{\pi\in \Pi}\E[v_{DR}(x;\pi)]$, $r=\sup_{\pi\in \Pi} \sqrt{\E[v_{DR}(z;\pi)^2]}$, $\mu_n = \Theta\left( \kappa(r, \cF_{\Pi}) + r \sqrt{\frac{\log(1/\delta)}{n}}\right)$ and 
\begin{equation}
    V_2 = \sup_{\pi\in \Pi:\, \E[v_{DR}(x;\pi_*)-v_{DR}(x;\pi)]\leq \mu_n} \Var(v_{DR}(x; \pi)-v_{DR}(x; \pi_*)).
\end{equation}
Then the policy $\pi_2$ returned by the out-of-sample regularized ERM, satisfies w.p. $1-\delta$ over the randomness of $S$:
\begin{align}
    \E[v_{DR}(\pi_*) - v_{DR}(\pi_2)] =~& O\left( \kappa(\sqrt{V_2}, \cF_{\Pi}) +  \sqrt{\frac{V_2\, \log(1/\delta)}{n}}\right)
\end{align}
and expected regret $O\left( \kappa(\sqrt{V_2}, \cF_{\Pi}) +  \sqrt{\frac{V_2}{n}}\right)$.
\end{corollary}

To arrive at our final theorem, we also need to account for the difference between $\E[v_{DR}(x;\pi)]$ and $V(\pi)$. This difference essentially stems from the error in the nuisance estimates, which introduce an error in $\theta_{DR}(y, a, z)$, such that $\E[\theta_{DR}(y, a, z)\mid z]\neq \theta(z)$. However, we can invoke the orthogonality of the doubly robust estimator and the general theorem of \cite{foster2019orthogonal} on generalization bounds of orthogonal losses to get:
\begin{lemma}\label{lem:ortho-regret}
For any policy $\pi_0\in \Pi$, let $\hat{\pi}$ be the outcome of any possibly randomized algorithm that satisfies w.p. $1-\delta/2$ a regret bound on the doubly robust objective, i.e. $\E[v_{DR}(x; \pi_0)-v_{DR}(x;\hat{\pi})] \leq R_{n,\delta}$. Moreover, suppose that the nuisance estimates satisfy a mean-squared error bound
\begin{equation}
    \max\left\{\E[(\hat{\theta}(z)-\theta_0(z))^2], \E[\|\hat{\Sigma}(z)-\Sigma_0(z)\|_{Fro}^2]\right\} := \chi_{n}^2
\end{equation}
Then w.p. $1-\delta$ over the randomness of the policy sample:
\begin{equation}
    V(\pi_0) - V(\hat{\pi}) \leq O\left(R_{n,\delta} + \chi_{n}^2\right)
\end{equation}
\end{lemma}
\begin{proof}
By Lemma~\ref{lem:univ-orthogonality} we have that the loss function $-\E[v_{DR}(x; \pi)]$ is universally orthogonal as defined in \cite{foster2019orthogonal}. Moreover, the loss is smooth with respect to the outputs of the nuisance functions and hence the second order derivatives of the loss with respect to the outputs of the nuisance functions are bounded. Thus the lemma follows by Theorem~2 of \cite{foster2019orthogonal}.
\end{proof}

If we assume that the nuisance estimation algorithm guarantees that w.p. $1-\delta$, $\chi_n^2\leq h_{n,\delta}^2$ then observe that combining Corollary~\ref{cor:policy-regret} and Lemma~\ref{lem:ortho-regret}, we get that for any policy $\pi_0$, the policy $\pi_2$ of the out-of-sample regularized ERM satisfies, w.p. $1-\delta$:
\begin{equation*}
     V(\pi_0) - V(\pi_2) \leq  O\left( \kappa(\sqrt{V_2}, \cF_{\Pi}) +  \sqrt{\frac{V_2\, \log(1/\delta)}{n}} + h_{n,\delta}^2\right)
\end{equation*}
Similarly, if we assume that the nuisance esitmation algorithm satisfies $\E[\chi_{n}^2]\leq h_n^2$, then:
\begin{equation*}
     \E[V(\pi_0) - V(\pi_2)] \leq  O\left( \kappa(\sqrt{V_2}, \cF_{\Pi}) +  \sqrt{\frac{V_2\, \log(1/\delta)}{n}} + h_{n}^2\right)
\end{equation*}
We continue by proving the probabilistic regret bound of the theorem and the in-expectation bound follows analogously.

Finally, we need to account for the error introduced by the nuisance errors on the quantity $V_2$, so as to connect it with the semi-parametric efficiency variance of each policy, i.e.:
\begin{equation}
    \Var(v_{DR}^0(x;\pi))
\end{equation}
where $v_{DR}^0(x;\pi) = \ldot{\theta_{DR}^0(y, a, z)}{\phi(\pi(z), z)}$, and $\theta_{DR}^0(y, a, z)$ is the analogue of the doubly robust function, $\theta_{DR}(y, a, z)$, evaluated at the true nuisance functions. Moreover, we want our the ``regret slice'' to be with respect to the true regret, i.e. we want to depend on the variance of policies that satisfy:
\begin{equation}
    V(\pi_*^0) - V(\pi) := \E[v_{DR}^0(x; \pi_*^0) -v_{DR}^0(x; \pi)] \leq \mu_n' 
\end{equation}
where $\pi_*^0=\arg\max_{\pi\in \Pi} V(\pi)$. We prove such a property in the following lemma:

\begin{lemma}\label{lem:variance-cont}
Consider the setting of Corollary~\ref{cor:policy-regret}. Suppose that the mean squared error of the nuisance estimates is upper bounded w.p. $1-\delta$ by $h_{n,\delta}^2$ and let $\epsilon_n=\mu_n+h_{n,\delta}^2$. Then:
\begin{equation}
    V_2^0 = \sup_{\pi,\pi'\in \Pi_*(\epsilon_n)} \Var(v_{DR}^0(x; \pi)-v_{DR}^0(x; \pi')) 
\end{equation}
Then $V_2 \leq V_2^0 + O(h_{n,\delta})$.
\end{lemma}
\begin{proof}
First observe that by Lemma~\ref{lem:ortho-regret} with $\pi_0=\pi_*$ and $\hat{\pi}=\pi$ (for any $\pi\in \cF_{\Pi}^2$), we have that:
\begin{align*}
    \E[v_{DR}(\pi_*)-v_{DR}(\pi)] \leq~& \mu_{n} \implies V(\pi_*)-V(\pi)\leq \mu_{n} + O(h_{n,\delta}^2)
\end{align*}
Similarly if we let $\pi_*^0=\arg\max_{\pi\in \Pi} \E[v_{DR}^0(x;\pi)]:=V(\pi)$, then observe that by definition of $\pi_*$: $\E[v_{DR}(x; \pi_*^0) - v_{DR}(x;\pi_*)]\leq 0$. Thus applying again Lemma~\ref{lem:ortho-regret} with $\pi_0=\pi_*^0$ and $\hat{\pi}=\pi_*$:
\begin{align*}
    \E[v_{DR}(\pi_*^0)-v_{DR}(\pi_*)] \leq~& 0 \implies V(\pi_*^0)-V(\pi_*)\leq O(h_{n,\delta}^2)
\end{align*}

Let $\Pi_*^0(\epsilon)=\{\pi\in \Pi: V(\pi_*^0) - V(\pi)\leq \epsilon\}$ and let $\epsilon_n = O(\mu_{n} + h_{n,\delta}^2)$. Thus we have that:
\begin{align*}
    V_2 \leq \sup_{\pi\in \Pi_*(\epsilon_n)} \Var(v_{DR}(x; \pi)-v_{DR}(x; \pi_*))
\end{align*}
Moreover, observe that $\pi_*\in \Pi_*^0(\epsilon_n)$. Hence:
\begin{align*}
    V_2 \leq \sup_{\pi,\pi'\in \Pi_*(\epsilon_n)} \Var(v_{DR}(x; \pi)-v_{DR}(x; \pi'))
\end{align*}

Moreover, by Lipschitzness of $\theta_{DR}(y, a, z)$ on the output of the nuisance functions, we also have that for any $\pi,\pi'\in \Pi(\epsilon_n)$:
\begin{equation}
    \Var(v_{DR}(x; \pi)-v_{DR}(x; \pi')) \leq 
    \Var(v_{DR}^0(x; \pi)-v_{DR}^0(x; \pi')) + O(h_{n,\delta})
\end{equation}

Hence, if we denote with:
\begin{equation*}
    V_2^0 = \sup_{\pi,\pi'\in \Pi_*(\epsilon_n)} \Var(v_{DR}^0(x; \pi)-v_{DR}^0(x; \pi')) 
\end{equation*}
Then we conclude that:
\begin{equation*}
    V_2 = V_2^0 + O(h_{n,\delta})
\end{equation*}
\end{proof}

Invoking Lemma~\ref{lem:variance-cont} and the concavity of the entropy integral function we get:
\begin{equation}
    V(\pi_*^0) - V(\hat{\pi}) \leq  O\left( \kappa(\sqrt{V_2^0}, \cF_{\Pi}) +  \sqrt{\frac{V_2^0\, \log(1/\delta)}{n}} + h_{n,\delta}^2 + h_{n,\delta}\frac{1}{\sqrt{n}} \right)
\end{equation}
Since $h_{n,\delta}=o(1)$, the last term is of lower order. This concludes the proof of the main regret Theorem~\ref{thm:main-regret}.

\section{Review of Semi-parametric Efficiency Bounds}

In this section, we review the theory of semi-parametric efficiency bounds studied in \cite{newey1990semiparametric} and \cite{bickel1993efficient}.

\subsection{Definitions} \label{sec:def}

\begin{definition}[Mean Square Differentiability]
Let $f(x ; \eta)$ denote the probability density function of a random variable $x$ where $\eta \in H$ is a finite dimensional parameter.  $f(x ; \eta)^{1/2}$is $\mu$-mean square continuously differentiable with respect to $\eta$ on  $H$ with derivative $f_{\eta}(x ;\eta)$ if for each $\eta \in H$  $\int \norm{f_{\eta}(x; \eta)}^2 d\mu$ is finite, and for every $\eta_{i} \rightarrow \eta$ with $\int \norm{f_{\eta}(x ;\eta_{i}) - f_{\eta}(x ;\eta)}^2 d \mu \rightarrow 0$ 
\begin{align*}
\int \Big( f(x ; \eta_{i})^{1/2} - f(x ; \eta)^{1/2} - f_{\eta}(x; \eta)^{\prime}(\eta_i - \eta) \Big)^2 d \mu / \norm{\eta_{i} - \eta}^{2} \rightarrow 0 
\end{align*}

\end{definition}

\begin{definition}[Smoothness]
$f(x ; \eta)$ is smooth if (i) $\eta \in H$, $H$ is open; (ii) there is a measure $\mu$ dominating $f(x ; \eta)$ for $\eta  \in H$ such that $f(x ; \eta)$ is continuous on $H$ a.s. $\mu$ ; (iii) $f(x ; \eta)^{1/2}$ is mean square differentiable.
\end{definition}

\begin{definition}[Score and Information Matrix]
For smooth $f(x ; \eta)$ the score for $\eta$ is defined as
$$S_\eta (x;  \eta) := 2 \dfrac{ f_{\eta}(x ; \eta)}{ f(x ;  \eta)}$$
in the support of $x$ and the information matrix is 
$$ \mathcal{I}(\eta) = \int S_{\eta} S^{\prime}_{\eta} f(x ; \eta) d \mu. $$
\end{definition}

\begin{definition}[Regularity]
A likelihood function $f(x ; \eta)$, $\eta \in H$, is regular if it is smooth and information matrix is non-singular in $H$. The efficiency bound of a regular model is given by Cramer-Rao bound and equals $\mathcal{I}(\eta)^{-1}$.
\end{definition}

\begin{definition}[Linearity]
Define a set $\mathcal{T}$ to be linear if $a s_{1} + b s_{2} \in \mathcal{T}$ for all real scalars $a$ and $b$ and elements $s_{1}$ and $s_{2}$ of $\mathcal{T}$.
\end{definition}

\subsection{Derivation of the Efficiency Bound} \label{sec:eff_bound}

Let data $(x_{1}, \dots, x_{n})$ consist of i.i.d copies of the random vector $(y,a,z)$. A semi-parametric model consists of a parameter vector $\alpha$ and a set of restrictions on the joint behavior of observables. In our model, the restrictions are given by the first order conditions of the linear projection
\begin{align*}
\E\left[\, (y - \langle \theta_0(z), \phi(a,z )\rangle) \phi(a,z)\mid z\right] &= 0 
\end{align*}
and the parameter is
\begin{align*} 
    \alpha = \int  \langle \theta (z),\phi(\pi(z),z) \rangle f(z) dz
\end{align*}
where $f(z)$ denotes the probability distribution function of $z$. First, we provide the definition of a parametric submodel.

\begin{definition}[Parametric Submodel]
For estimators with i.i.d data, a parametric submodel corresponds to a parameter vector $\eta$ and a likelihood function $\ell(x | \eta)$ for a single observation that satisfies the semi-parametric restrictions.
\end{definition}

A parametric submodel is a subset of the model distributions satisfying the semi-parametric assumptions. The reason parametric submodels are useful in analyzing semi-parametric efficiency is that for parametric models, the Cramer-Rao bound gives the lower bound on the variance of estimators of a parameter under some regulatory conditions. Since semi-parametric models impose weaker restrictions than any parametric model, it is natural to expect that the asymptotic variance of a semi-parametric model is no smaller than the bound for the parametric model.

In a parametric submodel, our parameter of interest can be written as
\begin{align}  \label{eq:param}
    \alpha = \int  \langle \theta (z ; \eta),\phi(\pi(z),z) \rangle f(z; \eta) dz
\end{align}

Next, we define the semi-parametric efficient bounds.

\begin{definition}[Semi-parametric Efficiency Bound]
The semi-parametric efficiency bound of a semi-parametric estimator is defined as the supremum of the Cramer-Rao bounds for all regular parametric submodels.
\end{definition}

This definition is intuitive because any semi-parametric estimator that is consistent and asymptotically normal cannot have a lower variance than the supremum of Cramer-Rao bounds. The regulatory conditions defined in Section \ref{sec:def} guarantee that the Cramer-Rao bound is well-defined and gives an asymptotic efficiency bound.

To be able to obtain the Cramer-Rao bound for the parameter of interest under a parametric submodel, the parameter must be pathwise differentiable.

\begin{definition} [Pathwise Differentiability]
A parameter $\alpha$ is pathwise-differentiable if $\alpha(\eta)$ is differentiable for all smooth parametric submodels and there exists $q \times 1$ random vector $d$ such that $\E[d^{\prime} d]$ is finite and for all regular parametric submodels 
\begin{align*}
\frac{\partial \alpha(\eta_{0})}{\partial \eta} = \E[ d S^{\prime}_{\eta} ]    
\end{align*}
where $\eta_{0}$ denotes the true value of the parameter in the sense that  $\ell(x | \eta_{0})$ corresponds to the likelihood function that generates the data.
\end{definition}
Pathwise differentiability of a parameter is a weak condition because, by Riesz representation theorem, a parameter is pathwise-differentiable if it can be written as a functional that is mean-square continuous. From the definition of $\alpha$ in Equation (\ref{eq:param}) it is easy to see that $\alpha$ is pathwise-differentiable by Riesz representation theorem.

For a pathwise-differentiable parameter, the Cramer-Rao bound can be written as a function of the pathwise derivative using the Delta method.
\begin{align*}
    \text{Var}(\alpha(\eta_{0}))& = \frac{\partial \alpha(\eta_{0})}{\partial \eta}  (\E[S_{\eta}S_{\eta}^{\prime}])^{-1} \frac{\partial \alpha(\eta_{0})}{\partial \eta}^\prime \\
    &= \E[ d S^{\prime}_{\eta} ]   (\E[S_{\eta}S_{\eta}^{\prime}])^{-1} \E[S_{\eta} d^{\prime} ]  
\end{align*}
We can write $ \text{Var}(\alpha(\eta)) $ as a second moment of a random variable as follows
\begin{align*}
    \text{Var}(\alpha(\eta_{0})) & =  \E[ d S^{\prime}_{\eta} ]   (\E[S_{\eta}S_{\eta}^{\prime}])^{-1} \E[S_{\eta} d^{\prime} ]   \\
    & = \E \big[  \E[ d S^{\prime}_{\eta} ] (\E[S_{\eta}S_{\eta}^{\prime}])^{-1} S_{\eta}S_{\eta}^{\prime}(\E[S_{\eta}S_{\eta}^{\prime}])^{-1} \E[S_{\eta} d^{\prime} ]   \big] \\
    & = \E[ d_{\eta} d_{\eta}^{\prime} ] 
\end{align*}
Note that $d_{\eta}$ is mean-zero since
\begin{align*}
\E[d_{\eta}] &= \E \big[  \E[ d S^{\prime}_{\eta} ] (\E[S_{\eta}S_{\eta}^{\prime}])^{-1} S_{\eta} ] \big] \\ 
&=  \E[ d S^{\prime}_{\eta} ] (\E[S_{\eta}S_{\eta}^{\prime}])^{-1} \E[S_{\eta}] \\
&= 0
\end{align*}
This is useful because the Cramer-Rao bound of $\alpha$ under a parametric submodel equals the variance of $d_{\eta}$. Note further from the definition of $d_{\eta}$ that it is the linear projection of pathwise-derivative $d$ on score $S_{\eta}$. Therefore, the largest value of this projection can be obtained by considering the projection space as the scores corresponding to all parametric submodels. To formalize this, we next define the tangent set:
\begin{definition}[Tangent Set]
Define the tangent set $\mathcal{T}$ to be the mean square closure of all $q$-dimensional linear combinations of scores $S_{\eta}$ for smooth parametric submodels:
\begin{align*}
T = \{ s \in \mathbbm{R} : \E[ \Vert s \Vert^{2} ] \leq \infty , \quad \exists A_{j} S_{\eta j} \quad \text{with} \quad \lim\limits_{j \rightarrow \infty} \E[ \Vert s - A_{j} S_{\eta j} \Vert^{2} ] =0 \}     
\end{align*}
\end{definition}

The projection of $d$ on the tangent set should have a larger variance than any particular submodel, suggesting that the projection should give the semi-parametric efficiency bound. The mathematical meaning of this projection on the tangent set is a least-squares projection in a Hilbert space of random vectors. This projection is defined as:
\begin{align*}
\delta \in \mathcal{T}, \quad \E[ (d - \delta) s ] = 0 \quad \text{for all} \quad s\in \mathcal{T}
\end{align*}
If $\mathcal{T}$ is linear, then $\delta$ exists and unique. It is called the efficient score because it equals the efficient influence function in asymptotically linear estimators.
\begin{theorem}[\cite{newey1990semiparametric}, Theorem 3.1]
Suppose that the parameter is differentiable, $\mathcal{T}$ is linear, and  $\E[\delta \delta^{\prime}] $ is nonsingular, for the projection $\delta$ of $d$ on $\mathcal{T}$. Then semi-parametric efficiency bound equals $\E[\delta \delta^{\prime}]$.
\end{theorem}

\section{Proof of Theorem~\ref{thm:semi-param-efficiency}}

\begin{proof}

We follow the steps outlined in Section (\ref{sec:eff_bound}) to calculate the semi-parametric efficiency bound of the parameter of interest:
\begin{align} \label{eq:alpha_int}
\alpha :=\E[ \langle \theta_{0}(z), \phi(\pi(z),z)\rangle ] 
\end{align} 
Let $ f(y,a \mid z) $ and $f(z)$ denote the conditional distribution of $(y,a)$ given $z$ and the marginal distribution of $z$, respectively. The density of data $(y,a,z)$ is then equal to:
\begin{align*}
    f(y,a,z) = f(y,a \mid z) f(z)
\end{align*}
We consider a regular parametric submodel, parameterized by $\eta$, to calculate the pathwise derivative of $\alpha(\eta)$:
\begin{align*}
    f(y,a,z ; \eta) = f(y,a \mid z; \eta) f(z; \eta)
\end{align*}
The corresponding scores for the parametric submodel is given by:
\begin{align*}
s_{\eta}(y,a,z ; \eta) &=  s_{\eta}(y,a \mid z; \eta) + s_{\eta}(z; \eta)
\end{align*}
where $ s_{\eta}(y,a,z ; \eta) = 2 \dfrac{f_{\eta}(y,a,z ; \eta)}{f(y,a,z ; \eta)}$, and other scores are defined similarly.

Under the parametric submodel $\alpha$ can be written as:
\begin{align} \label{eq:path}
    \alpha(\eta) = \int  \langle \theta (z ; \eta),\phi(\pi(z),z) \rangle f(z ;\eta) dz
\end{align}
The first step in semi-parametric efficiency bound derivation is to show that $ \alpha(\eta) $ is pathwise differentiable, i.e. there exists $ d(y,a,z; \eta_{0}) $ such that
\begin{align*}
\frac{\partial \alpha(\eta)}{\partial \eta} = \E[d(y,a,z ; \eta) S_{\eta}(y,a,z ; \eta)  ]
\end{align*}
Let $\eta_{0}$ denote the true parameter value in the sense that $ f(y,a,z ; \eta_{0})$ corresponds to the density of the data. To show pathwise differentiability, we differentiate Equation (\ref{eq:path}) under the integral sign and evaluate at $\eta = \eta_{0}$:
\begin{align}
    \dfrac{\partial \alpha(\eta_{0})}{\partial \eta} &= \int \left \langle \dfrac{\partial \theta (z ; \eta_{0})}{\partial \eta},\phi(\pi(z),z)   \right \rangle f(z ;\eta_{0}) dz + \int  \langle \theta_{0} (z; \eta_{0}),\phi(\pi(z),z) \rangle \dfrac{\partial f(z ; \eta_{0})}{\partial \eta} dz \\
    & =  \E \left[ \left \langle \dfrac{\partial \theta (z ; \eta_{0})}{\partial \eta},\phi(\pi(z),z)   \right \rangle \right] + \E[  \langle \theta (z; \eta_{0}),\phi(\pi(z),z) \rangle s_{\eta}(z ; \eta_{0})  ] \label{eq:pathwise}
\end{align}

To calculate $\partial \theta(z; \eta_0) / \partial \eta$ inside the expectations we use the first order conditions of the linear projection:
\begin{align*}
\E\left[\, (y - \langle \theta_0(z), \phi(a,z)\rangle) \phi(a,z)\mid z\right] &= 0 \\
 \int (y - \langle {\theta}(z ; \eta_0), \phi(a,z)\rangle) \phi_{i}(a,z) f(y, a \mid z ; \eta_0 ) dy da &= 0
\end{align*}
Taking the derivative under the integral sign and evaluating at $\eta_0$ for all $i$:
\begin{align*}
\E \left [ \left \langle \dfrac{\partial \theta (z ; \eta_{0})}{\partial \eta},\phi(a,z) \phi(a,z)^{T}   \right \rangle \mid z \right] +   \E[ (y - \langle {\theta}(z ; \eta_0), \phi(a,z)\rangle)  \phi(a,z) s_{\eta}(y, a \mid z ,\eta_{0}) \mid z  ] =0
\end{align*}
Solving for   $\partial \theta(z; \eta_0) / \partial \eta$
\begin{align*}
\partial \theta(z; \eta_0) / \partial \eta = \E \left [ \Sigma(z)^{-1} \phi(a,z) (y - \langle {\theta}(z ; \eta_0), \phi(a,z)\rangle)  s_{\eta}(y, a \mid z ;\eta_{0}) \mid z  \right  ]
\end{align*}
Substituting this into Equation (\ref{eq:pathwise}):
\begin{align}  \label{eq:pathwise_der}
 \dfrac{\partial \alpha(\eta_{0})}{\partial \eta} &=  \E \left[ \left \langle  \Sigma_{0}(z)^{-1}_{} \phi(a,z) (y - \langle {\theta_0}(z), \phi(a,z)\rangle)  ,\phi(\pi(z),z) \rangle \right) s_{\eta}(y, a  \mid z ; \eta_{0}) \right]+  \\
 & \quad \quad \E[   \langle \theta_0 (z),\phi(\pi(z),z) \rangle s_{\eta}(z ; \eta_{0}) ]  \nonumber \\ 
&=  \E \left[  \left( \langle  \theta_0 (z) + \Sigma_{0}(z)^{-1} \phi(a,z) (y - \langle {\theta_0}(z), \phi(a,z)\rangle)  ,\phi(\pi(z),z) \rangle  - \alpha(\eta_{0}) \right) \left( s_{\eta}(y, a \mid z, \eta_{0}) + s_{\eta}(z ; \eta_{0})  \right) \right] \nonumber \\  
    &=  \E \left[ d(y,a,z; \eta_{0}) \left( s_{\eta}(y, a \mid  z ; \eta_{0}) + s_{\eta}(z ; \eta_{0})  \right) \right] \nonumber \\  
    &=  \E \left[ d(y,a,z; \eta_{0}) \left( s_{\eta}(y, a \mid  z;\eta_{0})  \right) \right] 
\end{align}
The second line follows because:
\begin{align*}
\E[ \ldot{\theta_0(z)}{\phi(\pi(z),z)}  s_{\eta}(y, a \mid z, \eta_{0}) ] = \E[ \ldot{\theta_0(z)}{\phi(\pi(z),z)} \E[ s_{\eta}(y, a \mid z, \eta_{0}) \mid z ]] = 0
\end{align*}
\begin{align*}
\E[\alpha(\eta_{0})  s_{\eta}(z ; \eta_{0}) ] =  \alpha(\eta_{0})\E[  s_{\eta}(z ; \eta_{0}) ] = 0
\end{align*}
and 
\begin{align*}
\E [ \langle \Sigma_{0}(z)^{-1} \phi(a,z) (y - \langle {\theta_0}(z), \phi(a,z)\rangle)  ,\phi(\pi(z),z) \rangle   s_{\eta}(z ; \eta_{0}) ]  =0
\end{align*}
Subtracting $\alpha(\eta_{0})$ in the second line makes the pathwise derivative mean zero, which will prove useful later when projecting $ d(y,a,z; \eta_{0})$ on the tangent set.

Since Equation (\ref{eq:pathwise_der}) satisfies the condition given in the defition of pathwise differentiability, the pathwise derivative of $\alpha(\eta)$ is:
\begin{align*} 
d(y,a,z; \eta_{0}) = \left( \langle  \theta_0 (z) + \Sigma_{0}(z)^{-1} \phi(a,z) (y - \langle {\theta_0}(z), \phi(a,z)\rangle)  ,\phi(\pi(z),z) \rangle  - \alpha \right)
\end{align*}
The semi-parametric efficiency bound for $\alpha$ is the variance of the projection of $d(y,a,z ; \eta_{0})$ onto the tangent space defined as the closed linear span of the scores:
\begin{align*}
\mathcal{T} = \{ s(y,a \mid z) + s(z) \}    
\end{align*}
Note that the joint distribution is unrestricted so the only restrictions on the score functions are $E[s(y,x \mid z) \mid z ] =0 $ and \ $ E[s(z)] = 0 $ and they are smooth.

Next, we show that the pathwise derivative is already in the tangent set $ d(y,a,z; \eta_{0}) \in \mathcal{T}$. To see this we can write $d(y,a,z;\eta_{0}) $ as the sum of two functions:
\begin{align*} 
d(y,a,z; \eta_{0}) &= \big( \Sigma_{0}(z)^{-1} \phi(a,z) (y - \langle {\theta_0}(z), \phi(a,z)\rangle)  ,\phi(\pi(z),z) \rangle \big) + \big( \langle \theta_0(z), \phi(\pi(z),z) \rangle - \alpha \big) 
\end{align*}
The first component is mean independent of $z$:
\begin{align*}
\E[  \langle  \Sigma_{0}(z)^{-1} \phi(a,z) (y - \langle {\theta_0}(z), \phi(a,z)\rangle)  ,\phi(\pi(z),z) \rangle  \mid z] &=0
\end{align*}
The second component is function of only $z$ and has zero mean:
\begin{align*}
\E[ \langle \theta_0(z), \phi(\pi(z),z) \rangle - \alpha] & = 0 
\end{align*}
Therefore, the pathwise derivative equals the sum of two functions that satisfy the restrictions on score functions in the tangent set, namely, $E[s(y,x \mid z) \mid z ] =0 $ and \ $ E[s(z)] = 0 $. From this, we conclude that $d(y,a,z; \eta_{0})$ is in the tangent set; so the projection of $ d(y,a,z; \eta_{0})$ onto  $\mathcal{T}$ is equal to itself.

Therefore, the efficiency bound for $\alpha$ is:
\begin{align*}
V_{eff}(\alpha) &= Var(d(y,a,z;\eta_0)) \\
&=    Var(v_{DR}(y,a,z; \pi))
\end{align*}
Therefore, the doubly robust estimator, $v_{DR}(y,a,z; \pi)$, achieves the semi-parametric efficiency bound. This result extends to the difference of value functions by linearity of pathwise derivative.

To investigate the semi-parametric efficiency bound under the correct specification we use a result from \cite{chamberlain1992efficiency} who shows that under the correct specification the efficiency bound is:
\begin{align*}
V_{eff}^{c}(\alpha) &= Var \big( \langle  \theta_0 (z), \phi(\pi(z),z) \rangle \big) \\
&  \quad \quad + \E \big[ \phi(\pi(z),z) \E[ \phi(a,z)  \E[ \epsilon^2 \mid a,z ]^{-1}  \phi(a,z)^\prime \mid z ]^{-1} \phi(\pi(z),z)^{T}  \big]
\end{align*}
where $\epsilon =  (y - \langle {\theta_0}(z), \phi(a,z)\rangle) $ is defined as residuals.

Under the homoskedastivity assumption, $\E[\epsilon^2 \mid a,z ] = \sigma^{2}$, this efficiency bound becomes:
\begin{align*}
V_{eff}^{c}(\alpha) &= Var \big( \langle  \theta_0 (z), \phi(\pi(z),z) \rangle \big) + \\ 
&  \quad \quad \sigma^2 \E \big[ \phi(\pi(z),z) \E[ \phi(a,z) \phi(a,z)^\prime \mid z ]^{-1} \phi(\pi(z),z)^{T}  \big] \\
&= Var \big( \langle  \theta_0 (z), \phi(\pi(z),z) \rangle \big) + \sigma^2 \E \big[ \phi(\pi(z),z)  \Sigma_{0}(z)^{-1} \phi(\pi(z),z)^{T}  \big] 
\end{align*}
which is equal to the variance of the doubly robust estimator:
\begin{align*}
V_{eff}(\alpha) &= Var \big( \langle  \theta_0 (z), \phi(\pi(z),z) \rangle \big) +  \\
  &  \quad \quad \E \big[ \phi(\pi(z),z) \E[  \Sigma_{0}(z)^{-1} \phi(a,z)  \epsilon^2 \phi(a,z)^\prime \Sigma(z)^{-1} \mid z ] \phi(\pi(z),z)^{T}  \big] \\
&= Var \big( \langle  \theta_0 (z), \phi(\pi(z),z) \rangle \big) +  \\
&  \quad \quad  \sigma^2 \E \big[ \phi(\pi(z),z)  \Sigma_{0}(z)^{-1}  \E[  \phi(a,z) \phi(a,z)^\prime \mid z ] \Sigma(z)^{-1}  \phi(\pi(z),z)^{T}  \big] \\
&= Var \big( \langle  \theta_0 (z), \phi(\pi(z),z) \rangle \big) +  \\
&  \quad \quad \sigma^2 \E \big[ \phi(\pi(z),z)  \Sigma_{0}(z)^{-1}\Sigma_{0}(z)^{-1}\Sigma_{0}(z)^{-1}  \phi(\pi(z),z)^{T}  \big] \\
&= Var \big( \langle  \theta_0 (z), \phi(\pi(z),z) \rangle \big) + \sigma^2\E \big[ \phi(\pi(z),z)  \Sigma_{0}(z)^{-1} \phi(\pi(z),z)^{T}  \big] \\
&= Var(v_{DR}(y,a,z; \pi)) 
\end{align*}
\end{proof}

\section{Double Robustness Property of Policy Estimator}\label{app:dr}

\begin{theorem}[Double Robustness]\label{thm:double-robustness}
$V_{DR}(\pi)$ is an unbiased estimate of $V_0(\pi(z),z)$ if for all $z$, either $\E_{S_1\sim D^{n/2}}[\hat{\theta}(z)]=\theta_0(z)$ or $\E_{S_1\sim D^{n/2}}[\hat{\Sigma}(z)^{-1}] =\Sigma_{0}(z)^{-1}$, where expectation is taken over the randomness of the nuisance estimation sample $S_1$.
\end{theorem}
\begin{proof}
Let $\bar{\theta}(z) = \E_{S_1\sim D^{n/2}}[\hat{\theta}(z)]$ and $\bar{\Sigma}^{-1}(z) = \E_{S_1\sim D^{n/2}}[\hat{\Sigma}(z)^{-1}]$, be the expected value of the estimates at any input $z$, where the expectation is with respect to the randomness on the half-split of $n/2$ samples that were used for training the estimates. Due to sample-splitting and cross-fitting, the expected value of the doubly robust policy estimate can be written as:
\begin{align}
    \E[V_{DR}(\pi)] =~& \E\left[ \left\langle \bar{\theta}(z) + \bar{\Sigma}(z)^{-1}\, \phi(a,z)\, (y - \langle \bar{\theta}(z), \phi(a,z)\rangle), \phi(\pi(z),z) \right\rangle \right]
\end{align}
where the random variables $(y, a, z)$ are a fresh independent draw of the data generating process that generated the observational data. 

Observe that $y$ is an unbiased estimate of $V(a, z)$ conditional on $z$. Moreover, since $\theta_0(z)$ is
the minimizer of the conditional squared loss, taking the first order condition implies:
\begin{align*}
    \E[(V_0(a, z)  &- \ldot{\theta_0(z)}{\phi(a,z)})\phi(a,z)\mid z]=0  \Longleftrightarrow \\  \E[y\, & \phi(a,z)\mid z] = \E[\ldot{\theta_0(z)}{\phi(a,z)})\, \phi(a,z)\mid z]
\end{align*}
Thus we can re-write the expected value of the doubly robust policy estimate as:
\begin{align*}
    \E[V_{DR}(\pi)] =~& \E\left[ \left\langle \bar{\theta}(z) + \bar{\Sigma}(z)^{-1}\, \phi(a,z)\, (Y - \langle \bar{\theta}(z), \phi(a,z)\rangle), \phi(\pi(z),z) \right\rangle \right]\\
    =~& \E\left[ \left\langle \bar{\theta}(z) + \bar{\Sigma}(z)^{-1}\, \phi(a,z)\, \langle \theta_0(z) - \bar{\theta}(z), \phi(a,z)\rangle, \phi(\pi(z),z) \right\rangle \right]\\
    =~& \E\left[ \left\langle \bar{\theta}(z) + \bar{\Sigma}(z)^{-1}\, \phi(a,z) \phi(a,z)^T (\theta_0(z) - \bar{\theta}(z)), \phi(\pi(z),z) \right\rangle \right]\\
    =~& \E\left[ \left\langle \bar{\theta}(z) + \bar{\Sigma}(z)^{-1}\, \E[\phi(a,z) \phi(a,z)^T\mid z] (\theta_0(z) - \bar{\theta}(z)), \phi(\pi(z),z) \right\rangle \right]\\
    =~& \E\left[ \left\langle \bar{\theta}(z) + \bar{\Sigma}(z)^{-1}\, \Sigma_0(z) (\theta_0(z) - \bar{\theta}(z)), \phi(\pi(z),z) \right\rangle \right]\\
    =~& \E\left[ \left\langle \bar{\theta}(z) + \bar{\Sigma}(z)^{-1}\, \Sigma_0(z) (\theta_0(z) - \bar{\theta}(z)), \phi(\pi(z),z) \right\rangle \right]
\end{align*}
Hence we have:
\begin{align*}
\E[V_{DR}(\pi)] - V_0(\pi) =~& \E\left[ \left\langle \left(\bar{\Sigma}(z)^{-1}\, \Sigma_0(z) - \eye\right)\, \left(\theta_0(z) - \bar{\theta}(z)\right), \phi(\pi(z),z) \right\rangle \right]
\end{align*}
The right hand side is zero if either $\bar{\theta}(z)=\theta_0(z)$ or if $\bar{\Sigma}(z)^{-1}=\Sigma_0(z)$.
\end{proof}

\section{Lipschitz Variogram Settings and Binary Treatment}
\label{app:variogram}
 
For simplicity of notation, we let $v(x; \pi)=v_{DR}^0(x;\pi)$ and $\pi_*=\pi_*^0$ throughout this section, as the results are not specific to the doubly robust value function. Suppose that the value function of the policy learning problem has the following self-bounded Lipschitz property:
\begin{equation*}
    \Var(v(x; \pi)) - C\,\Var(v(x; \pi_*)) \leq L\, \left|\E[v(x;\pi)] - \E[v(x;\pi_*)]\right| = L(V(\pi_*) - V(\pi))
\end{equation*}
for some constants $C,L$, i.e. if a policy has value close to the optimal policy, the it does not have much larger variance.
Then we have that:
\begin{align*}
    \sup_{\pi,\pi' \in \Pi_*(\epsilon_n)} \Var(v(x; \pi)-v(x; \pi))\leq~& \sup_{\pi\in \Pi_*(\epsilon_n)} 4\,\Var(v(x; \pi)) \\
    \leq~& 4C\,\Var(v(x; \pi_*)) + 4\,L\, \sup_{\pi\in \Pi_*(\epsilon_n)}(V(\pi) - V(\pi_*))\\
    \leq~& 4C\,\underbrace{\Var(v(x; \pi_*))}_{V_*} + 4\,L\, \epsilon_n
\end{align*}
Thus we get regret rates of the form:
\begin{align*}
V(\pi_*) - V(\pi_2)
=~& O\left(\kappa(2\sqrt{C\,V_*}, \cF_{\Pi}) +  \sqrt{\frac{V_*\, \log(1/\delta)}{n}} + \epsilon_n \frac{1}{\sqrt{n}}\right)\\
=~&  O\left(\kappa(2\sqrt{C\,V_*}, \cF_{\Pi}) +  \sqrt{\frac{V_*\, \log(1/\delta)}{n}}\right)
\end{align*}
since $\epsilon_n=o(1)$.

\begin{example}[Binary Treatment]
In the case of binary treatment, the loss took the form:
\begin{equation}
    v(x;\pi) = \Gamma(z) \cdot (2\pi(z) - 1)
\end{equation}
with $\pi: Z\rightarrow \{0, 1\}$.
In this case observe that the self-bounded property is satisfied since:
\begin{align*}
    \Var(v(x;\pi)) =~& \E[v(x;\pi)^2] - \E[v(x;\pi)]^2\\
    =~& \E[\Gamma(z)^2 (2\pi(z)-1)^2] - V(\pi)^2\\
    =~& \E[\Gamma(z)^2] - V(\pi)^2
\end{align*}
Where the latter property holds since $(2\pi(z)-1)^2=1$ irrespective of $\pi(z)$. Thus the first part in the variance is independent of the policy, which is the crucial special property of the binary treatment case. This leads to the fact that:
\begin{equation}
    \Var(v(x;\pi)) - \Var(v(x;\pi_*)) = V(\pi_*)^2 - V(\pi)^2  \leq 2\, \left| V(\pi) - V(\pi_*)\right|
\end{equation}
Hence, the self-boundedness property holds with $C=1$ and $L=2$. Thus for the binary treatment setting we can achieve a regret rate whose leading term only depends on the semi-parametric efficient variance of the optimal policy.

As a concrete example, consider the case when the class $\cF_{\Pi}$ is a VC-subgraph class of VC dimension $d$, and let $S_n=\E_n[\sup_{\pi} v(x;\pi)^2]= \E_n[\Gamma(z)^2]$. Then Theorem 2.6.7 of \cite{VanDerVaartWe96} shows that: $\cH_2(\epsilon, \cF_{\Pi}, n)=O( d(1+\log(S_n/\epsilon)))$. This implies that
\begin{align*}
\kappa(r, \cF_{\Pi}) = O\left(\int_0^{r} \sqrt{d(1+\log(S_n/\epsilon))}d\epsilon\right) = O\left(r \sqrt{d}\sqrt{1+\log(S/r)}\right).
\end{align*}
Moreover, by Markov's inequality w.p. $1-\delta$, $S_n \leq \E[S_n]/\delta = \E[\sup_{\pi} v(x;\pi)^2]/\delta=\E[\Gamma(z)^2]/\delta:=S/\delta$. Hence, we can conclude that w.p. $1-\delta$:
\begin{align*}
V(\pi_*) - V(\pi_2) = O\left( r \sqrt{1+\log(S/r)}\sqrt{\frac{d}{n}}  + r \sqrt{\frac{\log(1/\delta)}{n}} + \frac{d(1+\log(S/r))}{n} \frac{\log(1/\delta)}{n}\right).
\end{align*}
Combining all the above we get a bound of the form (excluding lower order terms):
\begin{align*}
V(\pi_*) - V(\pi_2) = O\left(  \sqrt{V_*(1+\log(S/V_*))}\sqrt{\frac{d}{n}} +  \sqrt{\frac{V_*\, \log(1/\delta)}{n}}\right).
\end{align*}
which  recovers the result of \cite{athey2017efficient} for binary treatments up to constants.
\end{example}

\section{Doubly Robust Estimators in Pricing Experiment} \label{sec:pricing}

\subsection{Linear Model}

We want to estimate some regression models of $a(z)$ and $b(z)$ in the demand model. For instance, if these fall in some high-dimensional linear function class, we can estimate a regression between demand and the linear function class. Moreover, we need to estimate the covariance matrix, which in this case takes the simple form:
\begin{equation}
    \Sigma_0(z) = \begin{bmatrix} 
1 & \E[p\mid z] \\
\E[p\mid z] & \E[p^2\mid z] 
\end{bmatrix}
\end{equation}
whose inverse takes the form:
\begin{equation}
    \Sigma_0(z)^{-1} = \frac{1}{\Var(p\mid z)}\begin{bmatrix} 
\E[p^2\mid z] & -\E[p\mid z] \\
-\E[p\mid z] &  1
\end{bmatrix}
\end{equation}
If for instance the observational policy was homoskedastic (i.e. the exploration component was independent of the context $z$), then $\Var(p\mid z)$ is a constant $\sigma^2$ independent of $z$. Moreover, we can write:
\begin{equation}
    \E[p^2\mid z] = \sigma^2 + \E[p\mid z]^2
\end{equation}

Thus we only need to estimate the mean treatment policy $g(z)=\E[p\mid z]$ and the variance $\sigma^2$. Then the doubly robust estimate of $a(z)$ takes the form:
\begin{align*}
    a_{DR}(z) 
    =~& \hat{a}(z) + \left(1 + \hat{g}(z)\frac{\hat{g}(z) - p}{\hat{\sigma}^2} \right)\, (d - \hat{a}(z) - \hat{b}(z)\, p) \\
    b_{DR}(z) =~& \hat{b}(z) + \frac{p - \hat{g}(z)}{\hat{\sigma}^2} (d - \hat{a}(z) - \hat{b}(z)\, p)
\end{align*}

\subsection{Quadratic Model}

In the case where we observe the revenue our model becomes quadratic in prices
\begin{align*}
r &= a(z) x -   b (z) x^2 + \epsilon
\end{align*}  
The covariance matrix takes the form:
\begin{equation*}
    \Sigma_0(z) = \begin{bmatrix} 
\E[p^2 \mid z]  & \E[p^3 \mid z] \\
\E[p^3 \mid z] & \E[p^4 \mid z] 
\end{bmatrix}
\end{equation*}
whose inverse is:
\begin{equation*}
    \Sigma_0(z)^{-1} = \frac{1}{\E[p^4 \mid z]\E[p^2 \mid z] - \E[p^3 \mid z]^3 }\begin{bmatrix} 
\E[p^4\mid z] & -\E[p^3 \mid z] \\
-\E[p^3\mid z] & \E[p^2 \mid z]
\end{bmatrix}
\end{equation*}
Let $\mu_k(z)$ denote  $E[p^k \mid z]$. If the observational policy was homoskedastic and none of the central moments of price depends on $z$, using the recursive structure, the nuisance functions in the covariance matrix can be written as
\begin{align*}
    \mu_2(z) &= \mu^c_2 + \mu_1(z)^2 \\
    \mu_3(z) &= \mu^c_3 + 3\mu_2(z)\mu_1(z) - 2\mu_1(z)^3  \\
    \mu_4(z) &= \mu^c_4 + 4\mu_3(z)\mu_1(z)  - 6\mu_1(z)\mu_2(z) + 3 \mu_1(z)^4
\end{align*}
where $\mu^c_k$ denotes the k-th central moment of $p$. Therefore, we only need to estimate the mean treatment policy $\mu_1(z)$ and the central moments $\mu^c_2$, $\mu^c_3$ and $\mu^c_4$. Then, the doubly robust estimate of $a(z)$ and $b(z)$ take the form:
\begin{align*}
    a_{DR}(z) 
    =~& \hat{a}(z) + \left( \dfrac{\mu_4(z) p - \mu_3(z) p^2}{ \mu_4(z) \mu_2(z) - \mu_3(z)^2}\right)\, (d - \hat{a}(z)p - \hat{b}(z)\, p^2) \\
    b_{DR}(z) =~& \hat{b}(z) +  \left( \dfrac{\mu_2(z) p^2 - \mu_3(z) p}{ \mu_4(z) \mu_2(z) - \mu_3(z)^2}\right) (d - \hat{a}(z)p - \hat{b}(z)\, p^2)
\end{align*}

\section{Additional Experiment Results} \label{sec:figures}

\begin{figure}[H]
\centering
\subfloat[Policy Evaluation]{\includegraphics[width=0.9 \textwidth,
height=5cm]{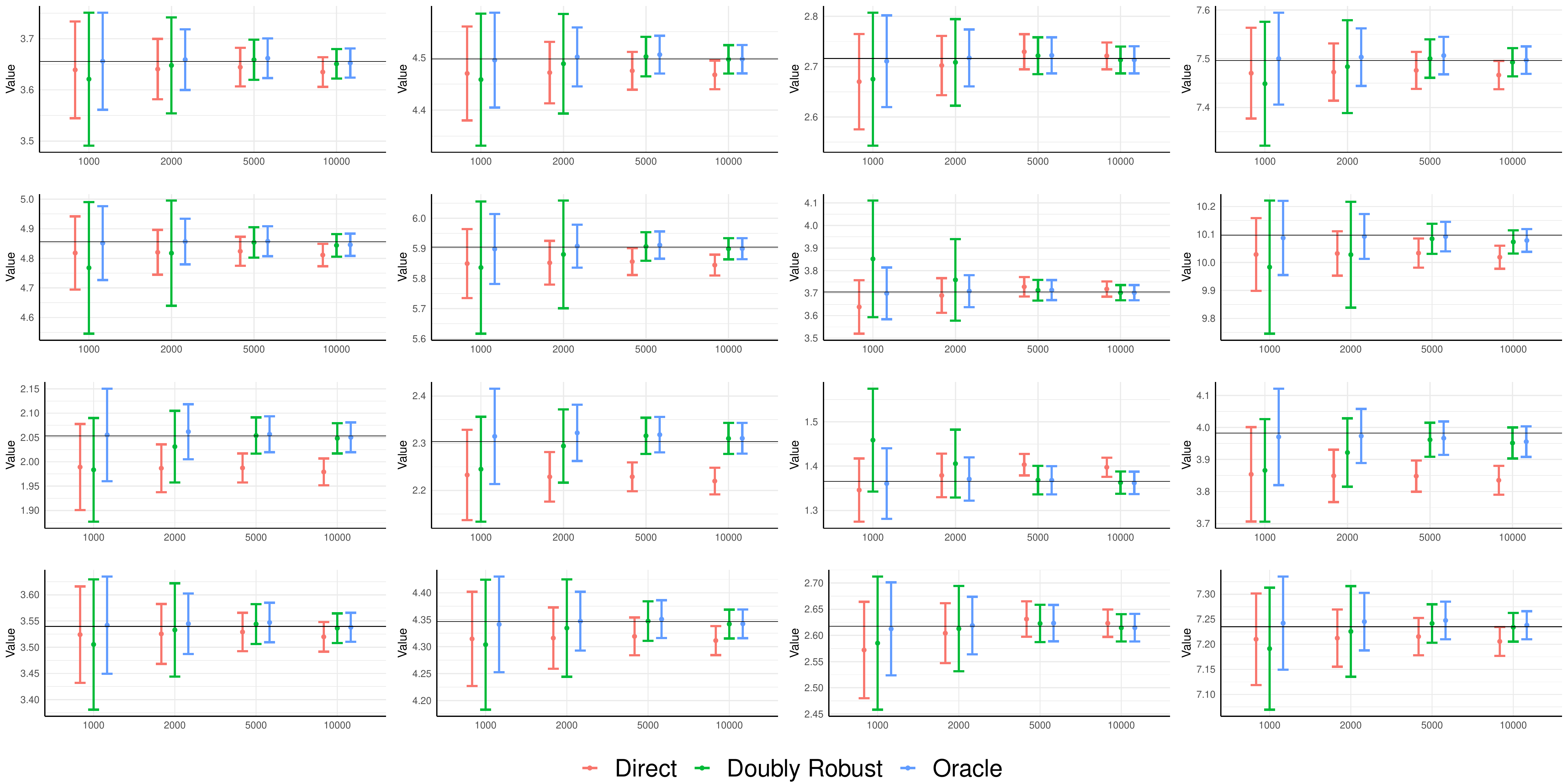}}

\subfloat[Regret]{\includegraphics[width=0.9 \textwidth, height=2.5cm]{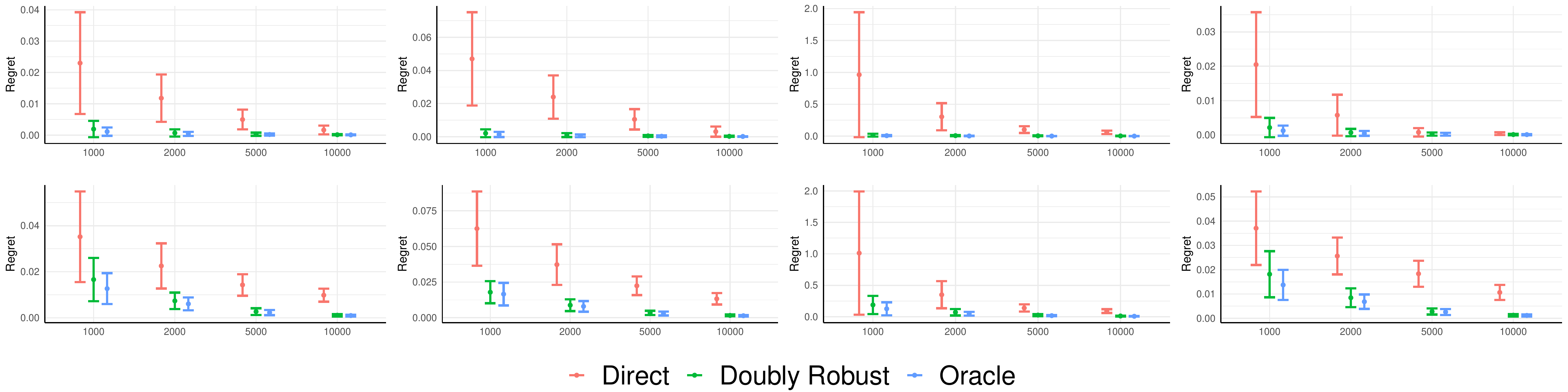}}
\caption{Linear, High Dimensional Regime: (a) Black line shows the true value of the policy, and each line shows the mean and standard deviation of the policy over 100 simulations. (b) each line shows the mean and standard deviation of the value of the corresponding policy over 100 simulations. We omit the results for the inverse propensity score method since they are too large to report together with the other estimates in the high dimensional regime.}
\label{fig:exp2} 
\end{figure}

\begin{figure}[t]
\centering
\subfloat[Policy Evaluation]{\includegraphics[width=0.9 \textwidth,
height=5cm]{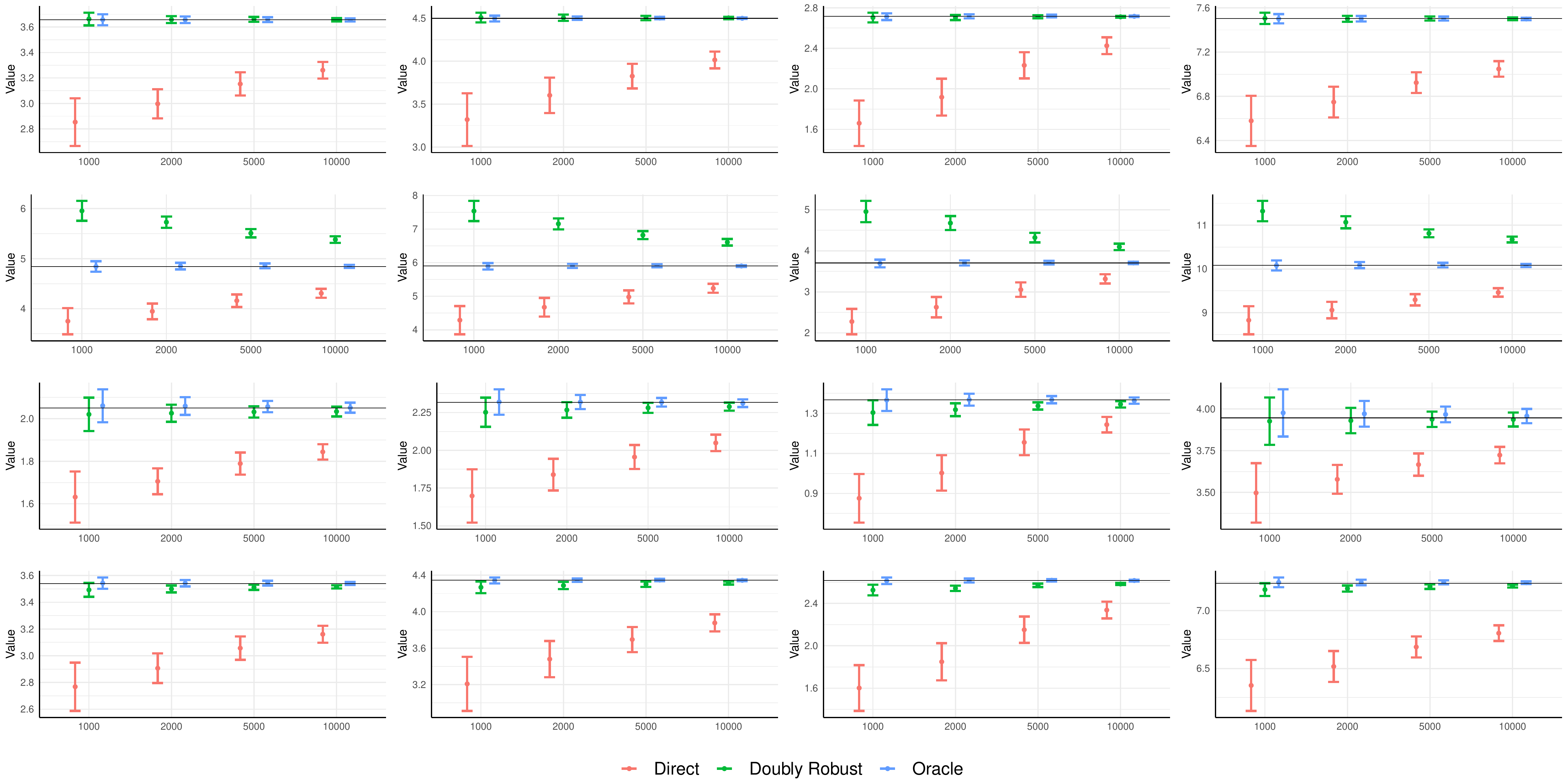}}

\caption{Quadratic, High Dimensional Regime: (a) Black line shows the true value of the policy, and each line shows the mean and standard deviation of the policy over 100 simulations. (b) each line shows the mean and standard deviation of the value of the corresponding policy over 100 simulations. We omit the results for the inverse propensity score method since they are too large to report together with the other estimates in the high dimensional regime.}
\label{fig:exp4} 
\end{figure}

\end{document}